%% file: htm.tex
\documentclass[a4paper,pdftex,UKenglish,10pt]{article}
\pdfoutput=1

 \usepackage{enumitem}
\setlist{nolistsep}
\usepackage{microtype}

\usepackage{multicol}
\usepackage[margin=1in]{geometry}
\usepackage{amsmath}
\usepackage{amsthm}
\usepackage{amssymb}
\usepackage{framed}
\usepackage{subcaption}
\usepackage{graphicx}
\usepackage{float}
\floatstyle{boxed} 
\restylefloat{figure}
\usepackage{tikz}
\usepackage{tabularx}
\usepackage{hyperref}
\usepackage{cleveref}
\crefname{section}{\S}{\S\S}
\usepackage{url}
\usepackage{macros}
\usepackage[ruled]{algorithm}
\usepackage{algpseudocode}
\usepackage{ioa_code}

\newcommand{\comnospace}{$\triangleright$}
\newcommand{\com}{\comnospace\ }
\usepackage{tabto}
\newcommand{\medcom}[1]{\tabto{5cm} \com \mbox{#1}}

\usepackage{listings}
\lstdefinestyle{customc}{
    belowcaptionskip=1\baselineskip,
    breaklines=true,
    frame=L,	
    xleftmargin=\parindent,
    language=C,
    showstringspaces=false,
    escapeinside={//}{\^^M},
    basicstyle=\scriptsize\ttfamily,
    keywordstyle=\bfseries\color{green!40!black},
    commentstyle=\itshape\color{gray!60!black},
    identifierstyle=\color{blue!50!black},
    stringstyle=\color{orange},
    numbers=left,                    
    numbersep=4pt,                   
    numberstyle=\tiny\color{black},  
    otherkeywords={then,word,process_local,type,xbegin,xabort,xend}
}

\def\S{\ensuremath{\mathcal{S}}}

\def\TS{\mathit{TS}}
\def\shared{\mathit{shared}}
\def\exclusive{\mathit{exclusive}}

\newcommand{\RNum}[1]{\uppercase\expandafter{\romannumeral #1\relax}}

\newcommand{\false}{\mathit{false}}

\newcommand{\remove}[1]{}

\newcommand{\Wset}{\textit{Wset}}
\newcommand{\Rset}{\textit{Rset}}
\newcommand{\Dset}{\textit{Dset}}

\newcommand{\txns}{\textit{txns}}

\newcommand{\Read}{\textit{read}}
\newcommand{\Write}{\textit{write}}
\newcommand{\TryC}{\textit{tryC}}
\newcommand{\TryA}{\textit{tryA}}

\newtheorem{theorem}{Theorem}
\newenvironment{proofsketch}[1][Proof sketch]{\noindent\textbf{#1.} }{\hfill $\Box$\\[2mm]}
\newtheorem{claim}[theorem]{Claim}
\newtheorem{remark}[theorem]{Remark}

\newtheorem{lemma}[theorem]{Lemma}

\title{On the Cost of Concurrency in Hybrid Transactional Memory}
\author{
	Trevor Brown$^1$~~~Srivatsan Ravi$^2$ \\
	$^1$\normalsize Dept. of CS, University of Waterloo \\
	$^2$\normalsize Dept. of CS and Information Sciences Institute, University of Southern California
}

\begin{document}

\maketitle

\begin{abstract}
State-of-the-art \emph{software transactional memory (STM)} implementations achieve 
good performance by carefully avoiding the overhead of \emph{incremental validation}
(i.e., re-reading previously read data items to avoid inconsistency) while
still providing \emph{progressiveness} (allowing transactional aborts only due to \emph{data conflicts}).
Hardware transactional memory (HTM) implementations promise even better performance, 
but offer no progress guarantees.
Thus, they must be combined with STMs, leading to \emph{hybrid} TMs (HyTMs)
in which hardware transactions must be \emph{instrumented} (i.e., access metadata) 
to detect contention with software transactions.

We show that, unlike in progressive STMs, software transactions in progressive HyTMs
cannot avoid incremental validation.
In fact, this result holds even if hardware transactions can \emph{read} metadata 
\emph{non-speculatively}. 
We then present \emph{opaque} HyTM algorithms providing \emph{progressiveness for a subset of transactions} 
that are  optimal in terms of hardware instrumentation. 
We explore the concurrency vs. hardware instrumentation vs. software validation
trade-offs for these algorithms.
Our experiments with Intel and IBM POWER8 HTMs   
seem to suggest that (i) the \emph{cost of concurrency} also exists in practice, 
(ii) it is important to implement HyTMs that provide progressiveness for a maximal set of transactions without incurring high hardware instrumentation overhead or
using global contending bottlenecks and (iii) 
there is no easy way to derive more efficient HyTMs by taking advantage of non-speculative accesses within hardware.

\end{abstract}

\input{intro}
\input{model_alt}
\input{spec}
\input{algos}
%
\input{Evaluation}
%
\input{related}
\bibliography{references2}
\newpage
\appendix
\input{appendix-proof}
\input{mproof}

\end{document}

%% file: intro.tex
\section{Introduction}
\label{sec:intro}
The \emph{Transactional Memory (TM)} abstraction is a synchronization mechanism 
that allows the programmer to \emph{optimistically} execute sequences of shared-memory
operations as \emph{atomic transactions}.
Several software TM designs~\cite{norec, ST95,HLM+03,fraser} have been introduced subsequent to the original TM proposal based in
hardware~\cite{HM93}. 
The original dynamic STM implementation DSTM~\cite{HLM+03} ensures that a transaction aborts only if there is a read-write \emph{data conflict} with a concurrent
transaction (\`a la \emph{progressiveness}~\cite{tm-book}). However, to satisfy \emph{opacity}~\cite{tm-book}, read operations in DSTM must \emph{incrementally} validate
the responses of all previous read operations to avoid inconsistent executions. 
This results in quadratic  (in the size of the transaction's read
set) step-complexity for transactions. Subsequent STM 
implementations like NOrec~\cite{norec} and TL2~\cite{DSS06}
minimize the impact on performance due to incremental validation.
NOrec uses a global sequence lock that is read at the start of a transaction and performs \emph{value-based}
validation during read operations only if the value of the global lock has been changed (by an updating transaction) 
since reading it.
TL2, on the other hand, eliminates incremental validation completely.
Like NOrec, it uses a global sequence lock, but each data item also 
has an associated sequence lock value that is updated alongside the data item.
When a data item is read, if its associated sequence lock value is different 
from the value that was read from the sequence lock at the start of the transaction, then the transaction aborts.

In fact, STMs like TL2 and NOrec ensure progress in the absence of data conflicts with 
O(1) step complexity read operations and \emph{invisible reads} (read operations which 
do not modify shared memory).
Nonetheless, TM designs that are implemented entirely in software still incur significant performance overhead.
Thus, current CPUs have included instructions to mark a block of memory accesses as transactional~\cite{Rei12, asf, bluegene}, allowing them to be executed \emph{atomically} in hardware.
Hardware transactions promise better performance than STMs, but they offer no progress guarantees 
since they may experience \emph{spurious} aborts. This motivates the need for
\emph{hybrid} TMs in which the \emph{fast} hardware transactions are 
complemented with \emph{slower} software transactions that do not have spurious aborts.

To allow hardware transactions in a HyTM to detect conflicts with software transactions, hardware transactions must be \emph{instrumented} to perform additional metadata accesses, which introduces overhead.
Hardware transactions typically provide automatic conflict detection at cacheline granularity,
thus ensuring that a transaction will be aborted if it experiences memory contention on a cacheline.
This is at least the case with Intel's Transactional Synchronization Extensions~\cite{haswell}.
The IBM POWER8 architecture additionally allows hardware transactions to access metadata \emph{non-speculatively}, 
thus bypassing automatic conflict detection. While this has the advantage of potentially reducing contention aborts
in hardware, this makes the design of HyTM implementations potentially harder to prove correct.

In \cite{hytm14disc}, it was shown that hardware transactions in opaque progressive HyTMs must perform
at least one metadata access per transactional read and write.
In this paper, we show that in opaque progressive HyTMs with invisible reads, 
software transactions \textit{cannot} avoid incremental validation.
Specifically, we prove that \textit{each read operation} of a software transaction in a progressive HyTM
must necessarily incur a validation cost that is \emph{linear} 
in the size of the transaction's read set. 
This is in contrast to TL2 which is progressive and has constant complexity read operations.
Thus, in addition to the linear instrumentation cost in hardware transactions, there is a quadratic step complexity cost in software transactions.

We then present opaque HyTM algorithms providing \emph{progressiveness for a subset of transactions} that are  
optimal in terms of hardware instrumentation.
Algorithm~1 is progressive for all transactions, but it incurs high instrumentation overhead in practice.
Algorithm~2 avoids all instrumentation in fast-path read operations, but is progressive only for slow-path reading transactions.
We also sketch how \emph{some} hardware instrumentation can be performed \textit{non-speculatively} without violating opacity.

Extensive experiments were performed to characterize the \textit{cost of concurrency} in practice.
We studied the instrumentation-optimal algorithms, as well as TL2, Transactional Lock Elision (TLE)~\cite{tle} and Hybrid NOrec~\cite{hynorecriegel} on both Intel and IBM POWER architectures.
Each of the algorithms we studied contributes to an improved understanding of the concurrency vs. hardware instrumentation vs. software validation trade-offs for HyTMs.
Comparing results between the very different Intel and IBM POWER architectures also led to new insights.
%
%
Collectively, our results suggest 
the following.
(i) The \emph{cost of concurrency} is significant in practice; high hardware instrumentation impacts performance negatively on Intel and much more so on POWER8 due to its limited transactional cache capacity.
(ii) It is important to implement HyTMs that provide progressiveness for a maximal set of transactions without incurring high hardware instrumentation overhead or using global contending bottlenecks.
(iii) There is no easy way to derive more efficient HyTMs by taking advantage of non-speculative accesses supported within the fast-path in POWER8 processors. 

\vspace{1mm}\noindent\textbf{Roadmap.}
The rest of the paper is organized as follows.
\cref{sec:hytm} presents details of the HyTM model that extends the model introduced in \cite{hytm14disc}.
\cref{sec:lb} presents our main lower bound result on the step-complexity of slow-path transactions in progressive HyTMs
while \cref{sec:hytmalgos} presents opaque HyTMs that are progressive for a subset of transactions.
\cref{sec:eval} presents results from experiments on Intel Haswell and IBM POWER8 architectures which provide a clear characterization of the cost
of concurrency in HyTMs, and study the impact of 
non-speculative (or direct) accesses within hardware transactions on performance.
\cref{sec:rel} presents the related work along with concluding remarks. Formal proofs 
appear 
in the Appendix.
%

%% file: model_alt.tex
\section{Hybrid transactional memory (HyTM)}
\label{sec:hytm}
%
\vspace{1mm}\noindent\textbf{Transactional memory (TM).} 
A \emph{transaction} is a sequence of \emph{transactional operations}
(or \emph{t-operations}), reads and writes, performed on a set of \emph{transactional objects} 
(\emph{t-objects}). 
A TM \emph{implementation} provides a set of
concurrent \emph{processes} with deterministic algorithms that implement reads and
writes on t-objects using  a set of \emph{base objects}.

\vspace{1mm}\noindent\textbf{Configurations and executions.} 
A \emph{configuration} of a TM implementation specifies the state of each base object and each process. 
In the \emph{initial} configuration, each base object has its initial value and each process is in its initial state. 
An \emph{event} (or \emph{step}) of a transaction invoked by some process is an invocation of a t-operation, 
a response of a t-operation, or an atomic \emph{primitive} operation applied to base object along with its response. 
An \emph{execution fragment} is a (finite or infinite) sequence of events $E = e_1,e_2,\dots$. 
An \emph{execution} of a TM implementation $\mathcal{M}$ is an
execution fragment where, informally, each event respects the
specification of base objects and the algorithms specified by $\mathcal{M}$.

For any finite execution $E$ and execution fragment $E'$, $E\cdot E'$ denotes the concatenation of $E$ and $E'$,
and we say that $E\cdot E'$ is an \emph{extension}
of $E$.
For every transaction identifier $k$,
$E|k$ denotes the subsequence of $E$ restricted to events of
transaction $T_k$.
If $E|k$ is non-empty,
we say that $T_k$ \emph{participates} in $E$,
Let $\txns(E)$ denote the set of transactions that participate in $E$.
Two executions $E$ and $E'$
are \emph{indistinguishable} to a set $\mathcal{T}$ of transactions, if
for each transaction $T_k \in \mathcal{T}$, $E|k=E'|k$.
A transaction $T_k\in \txns(E)$ is \emph{complete in $E$} if
$E|k$ ends with a response event.
The execution $E$ is \emph{complete} if all transactions in $\txns(E)$
are complete in $E$.
A transaction $T_k\in \txns(E)$ is \emph{t-complete} if $E|k$
ends with $A_k$ or $C_k$; otherwise, $T_k$ is \emph{t-incomplete}.
We consider the dynamic programming model: the \emph{read set} (resp., the \emph{write set}) of a transaction $T_k$ in an execution $E$,
denoted $\Rset_E(T_k)$ (resp., $\Wset_E(T_k)$), is the set of t-objects that $T_k$ attempts to read (and resp. write) 
by issuing a t-read (resp., t-write) invocation in $E$ (for brevity, we sometimes 
omit the subscript $E$ from the notation).

We assume that base objects are accessed with \emph{read-modify-write} (rmw) primitives. 
A rmw primitive event on a base object is \emph{trivial} if, in any configuration, its application
does not change the state of the object. 
Otherwise, it is called \emph{nontrivial}.
Events $e$ and $e'$ of an execution $E$  \emph{contend} on a base
object $b$ if they are both primitives on $b$ in $E$ and at least 
one of them is nontrivial.

\vspace{1mm}\noindent\textbf{Hybrid transactional memory executions.}
We now describe the execution model of a \emph{Hybrid transactional memory (HyTM)} implementation.
In our HyTM model, shared memory configurations may be modified by accessing base objects via two kinds of
primitives: \emph{direct} and \emph{cached}.
(i) In a direct (also called non-speculative) access, the rmw primitive operates on the memory state:
the direct-access event atomically reads the value of the object in
the shared memory and, if necessary, modifies it.
(ii) In a cached access performed by a process $i$, the rmw primitive operates on the \emph{cached}
state recorded in process $i$'s \emph{tracking set} $\tau_i$.

More precisely, $\tau_i$ is a set of triples $(b, v, m)$ where $b$ is a base object identifier, $v$ is a value, 
and $m \in \{\shared, \exclusive\}$ is an access \emph{mode}. 
The triple $(b, v, m)$ is added to the tracking set when $i$ performs a cached
rmw access of $b$, where $m$ is set to $\exclusive$ if the access is
nontrivial, and to $\shared$ otherwise.  
We assume that there exists some constant $\TS$
such that the condition $|\tau_i| \leq \TS$ must always hold; this
condition will be enforced by our model.
A base object $b$ is \emph{present} in $\tau_i$ with mode $m$ if $\exists v, (b,v,m) \in \tau_i$.

%
\vspace{1mm}\noindent\textbf{Hardware aborts.}
A tracking set can be \emph{invalidated} by a concurrent process: 
if, in a configuration $C$ where  $(b,v,\exclusive)\in\tau_i$
(resp., $(b,v,\shared)\in\tau_i)$,  a process $j\neq i$ applies any primitive 
(resp., any \emph{nontrivial} primitive) to $b$, then $\tau_i$ becomes
\emph{invalid} and any subsequent event invoked by $i$
sets $\tau_i$ to $\emptyset$ and returns $\bot$. We refer to this event as a \emph{tracking set abort}.

Any transaction executed by a \emph{correct} process that performs at least one cached access must necessarily perform a \emph{cache-commit} primitive 
that determines the terminal response of the transaction.
A cache-commit primitive
issued by process $i$ with
a valid $\tau_i$ does the following: for each base object $b$ such that $(b,v,\exclusive) \in \tau_i$, the value of $b$ in $C$ is updated to $v$. 
Finally, $\tau_i$ is set to $\emptyset$ and the operation returns $\textit{commit}$. 
We assume that a fast-path transaction $T_k$ returns $A_k$
as soon a cached primitive or \emph{cache-commit} returns $\bot$.

\vspace{1mm}\noindent\textbf{Slow-path and fast-path transactions.}
We partition HyTM transactions into \emph{fast-path transactions} and \emph{slow-path transactions}.
A slow-path transaction models a regular software transaction.
An event of a slow-path transaction is either an invocation or response of a t-operation, or
a direct rmw primitive on a base object. 
A fast-path transaction essentially encapsulates a hardware transaction. Specifically, in any execution $E$,
we say that a transaction $T_k\in \ms{txns}(E)$ is a fast-path transaction if $E|k$ contains at least one cached event.
An event of a \emph{hardware transaction} is either an invocation or response of a t-operation, or
a direct trivial access or a cached access, or a cache-commit primitive.
\begin{remark}[Tracking set aborts]
\label{re:traborts}
Let $T_k \in \ms{txns}(E)$ be any t-incomplete fast-path transaction executed by process $i$, 
where $(b,v,\exclusive)\in\tau_i$ (resp., $(b,v,\shared)\in\tau_i$) after execution $E$, and $e$ be any event (resp., nontrivial event) 
that some process $j\neq i$ is poised to apply after $E$.
The next event of $T_k$ in any extension of $E\cdot e$ is $A_k$.
%
\end{remark}
\begin{remark}[Capacity aborts]
\label{re:capacity}
Any cached access performed by a process $i$ executing a fast-path 
transaction $T_k$; $|\Dset(T_k)|>1$ first checks the condition $|\tau_i|=\TS$, where $\TS$ is a pre-defined constant, and if so, it
sets $\tau_i=\emptyset$ and immediately returns $\bot$. 
\end{remark}

\vspace{1mm}\noindent\textbf{Direct reads within fast-path.}
Note that we specifically allow hardware transactions to perform reads without adding the corresponding base object to
the process' tracking set, thus modeling the \emph{suspend/resume} instructions supported by 
IBM POWER8 architectures. Note that Intel's HTM does not support this feature: an event of a fast-path transaction
does not include any direct access to base objects.

\vspace{1mm}\noindent\textbf{HyTM properties.}
We consider the TM-correctness property of \emph{opacity}~\cite{tm-book}: an execution
$E$ is opaque if there exists a \emph{legal} (every t-read of a t-object returns the value of its latest committed t-write) sequential execution $S$ equivalent to some t-completion of $E$
that respects the \emph{real-time ordering} of transactions in $E$.
We also assume a weak \emph{TM-liveness} property for t-operations: every t-operation returns a matching
response within a finite number of its own steps if running step-contention free from a configuration in which every other transaction is t-complete.
Moreover, we focus on HyTMs that provide \emph{invisible reads}: t-read operations do not perform
nontrivial primitives in any execution.
%
%
%
%

%% file: spec.tex
\section{Progressive HyTM must perform incremental validation}
\label{sec:lb}
In this section, we show that it is impossible to implement opaque \emph{progressive} HyTMs with \emph{invisible reads}
with $O$(1) step-complexity read operations for slow-path transactions. 
This result holds even if fast-path transactions may perform
direct trivial accesses.

Formally, we say that a HyTM implementation $\mathcal{M}$ is progressive
for a set $\mathcal{T}$ of transactions
if in any execution $E$ of $\mathcal{M}$; $\mathcal{T} \subseteq \ms{txns}(E)$, 
if any transaction $T_k \in \mathcal{T}$ returns $A_k$ in $E$, there exists 
another concurrent transaction $T_m$ that \emph{conflicts} (both access the same t-object and at least one writes) with $T_k$ in $E$~\cite{tm-book}.

We construct an execution of a progressive opaque HyTM in which every t-read performed by a read-only slow-path transaction
must access linear (in the size of the read set) number of distinct base objects.
\begin{theorem}
\label{th:impossibility}
Let $\mathcal{M}$ be any progressive opaque HyTM implementation providing invisible reads.
There exists an execution $E$ of $\mathcal{M}$ and some slow-path read-only transaction $T_k \in \ms{txns}(E)$
that incurs a time complexity of $\Omega (m^2)$; $m=|\Rset(T_k)|$.
\end{theorem}
\begin{proofsketch}
We construct an execution of a read-only slow-path transaction $T_{\phi}$ that performs $m \in \mathbb{N}$
distinct t-reads of t-objects $X_1,\ldots , X_m$. We show inductively that for each 
$i\in \{1,\ldots , m\}$; $m \in \mathbb{N}$, the $i^{\ms{th}}$ t-read must access $i-1$ distinct base objects
during its execution. The (partial) steps in our execution are depicted in Figure~\ref{fig:indis}.

For each $i\in \{1,\ldots , m\}$, $\mathcal{M}$ has an execution of the form depicted in Figure~\ref{sfig:inv-2}.
Start with the complete step contention-free execution of slow-path read-only transaction $T_{\phi}$ that performs
$(i-1)$ t-reads: $\Read_{\phi}(X_1)\cdots \Read_{\phi}(X_{i-1})$, followed by the t-complete step contention-free execution of a fast-path transaction $T_{i}$
that writes $nv_i\neq v_i$ to $X_i$ and commits and then the complete step contention-free execution fragment of $T_{\phi}$ that performs its $i^{th}$ t-read:
$\Read_{\phi}(X_i) \rightarrow nv_i$. Indeed, by progressiveness, $T_i$ cannot incur tracking set aborts and since it accesses only a single t-object, it cannot incur capacity aborts.
Moreover, in this execution, the t-read of $X_i$ by slow-path transaction $T_{\phi}$ must return the value $nv$ written by fast-path transaction $T_i$ since this execution is indistinguishable
to $T_{\phi}$ from the execution in Figure~\ref{sfig:inv-1}. 

We now construct $(i-1)$ different executions of the form depicted in Figure~\ref{sfig:inv-3}: for each $\ell \leq (i-1)$, 
a fast-path transaction $T_{\ell}$ (preceding $T_i$ in real-time ordering, but invoked following the $(i-1)$ t-reads by $T_{\phi}$) writes $nv_{\ell}\neq v$ to $X_{\ell}$ and commits, followed by
the t-read of $X_i$ by $T_{\phi}$. Observe that, $T_{\ell}$ and $T_i$ which access mutually disjoint data sets cannot contend on each other since if they did, they would concurrently contend
on some base object and incur a tracking set abort, thus violating progressiveness.
Indeed, by the TM-liveness property we assumed (cf. Section~\ref{sec:hytm}) and invisible reads for $T_{\phi}$, each of these $(i-1)$ executions exist. 

In each of these $(i-1)$ executions, the final t-read of $X_i$ cannot return the new value $nv$:
the only possible serialization for transactions is $T_{\ell}$, $T_i$, $T_{\phi}$; but the $\Read_{\phi}(X_{\ell})$
performed by $T_k$ that returns the initial value $v$ is not legal in this serialization---contradiction to the assumption of opacity.
In other words, slow-path transaction $T_{\phi}$ is forced to verify the validity of t-objects in $\Rset(T_{\phi})$.
Finally, we note that, for all $\ell, \ell' \leq (i-1)$;$\ell' \neq \ell$, 
fast-path transactions $T_{\ell}$ and $T_{\ell'}$ access mutually disjoint sets of base objects thus forcing the t-read of $X_i$ to access least $i-1$ different base objects
in the worst case.
Consequently, for all $i \in \{2,\ldots, m\}$, slow-path transaction $T_{\phi}$ must perform at least $i-1$ steps 
while executing the $i^{th}$ t-read in such an execution.
\end{proofsketch}

\noindent\textbf{How STM implementations mitigate the quadratic lower bound step complexity.}
NOrec~\cite{norec} is a progressive opaque STM that minimizes the average step-complexity resulting from incremental 
validation of t-reads. Transactions read a global versioned lock at the start, and perform value-based validation
during t-read operations \emph{iff} the global version has changed.
TL2~\cite{DSS06} improves over NOrec by circumventing the lower bound
of Theorem~\ref{th:impossibility}. Concretely, TL2 associates a global version with each t-object updated during
a transaction and performs validation with O(1) complexity during t-reads by simply verifying if the version
of the t-object is greater than the global version read at the start of the transaction. Technically,
NOrec and algorithms in this paper provide a stronger definition of progressiveness: a transaction may abort
only if there is a prefix in which it conflicts with another transaction and both are t-incomplete. TL2 on the other hand allows
a transaction to abort due to a concurrent conflicting transaction.

\vspace{1mm}\noindent\textbf{Implications for disjoint-access parallelism in HyTM.}
The property of disjoint-access parallelism (DAP), in its \emph{weakest} form, ensures that two transactions 
concurrently contend on the same base object 
only if their data 
sets are connected in the \emph{conflict graph}, capturing 
data-set overlaps among all concurrent transactions~\cite{AHM09}. It is well known that weak DAP STMs with invisible reads must perform incremental validation even if the required TM-progress condition requires
transactions to commit only in the absence of any concurrent transaction~\cite{tm-book,prog15-pact}. For example, DSTM~\cite{HLM+03} is a weak DAP STM that is progressive and consequently incurs the validation
complexity. On the other hand, TL2 and NOrec are not weak DAP since they employ a global versioned lock that mitigates the cost of incremental validation, but this allows two transactions accessing
disjoint data sets to concurrently contend on the same memory location. Indeed, this inspires the proof of Theorem~\ref{th:impossibility}. 
\begin{figure*}[!t]
\begin{center}
	\begin{subfigure}{\linewidth}{\scalebox{0.6}[0.6]{\input{dap}}}
	\caption{Slow-path transaction $T_{\phi}$ performs $i-1$ distinct t-reads (each returning the initial value) followed by the t-read of $X_i$ that returns value $nv$ 
	written by fast-path transaction $T_i$}\label{sfig:inv-1}
	\end{subfigure}
        \\
        \vspace{2mm}
	\begin{subfigure}{\linewidth}{\scalebox{0.6}[0.6]{\input{dap2}}}
	\caption{Fast-path transaction $T_i$ does not contend with any of the $i-1$ t-reads performed by $T_{\phi}$ and must be committed in this execution since it cannot incur a tracking set or capacity abort.
	The t-read of $X_i$ must return $nv$ because this execution is indistinguishable to $T_{\phi}$ from \ref{sfig:inv-1}}
	\label{sfig:inv-2} 
	\end{subfigure}
	\\
	\vspace{2mm}
	\begin{subfigure}{\linewidth}{\scalebox{0.6}[0.6]{\input{dap3}}}
	 \caption{In each of these each $i-1$ executions, fast-path transactions cannot incur a tracking set or capacity abort. By opacity, the t-read of $X_i$ by $T_{\phi}$ cannot return new value $nv$.
	Therefore, to distinguish the $i-1$ different executions, t-read of $X_i$ by slow-path transaction $T_{\phi}$ is forced
	to access $i-1$ different base objects}
	\label{sfig:inv-3}
	\end{subfigure}
	\caption{Proof steps for Theorem~\ref{th:impossibility}
        \label{fig:indis}} 
\end{center}
\end{figure*}
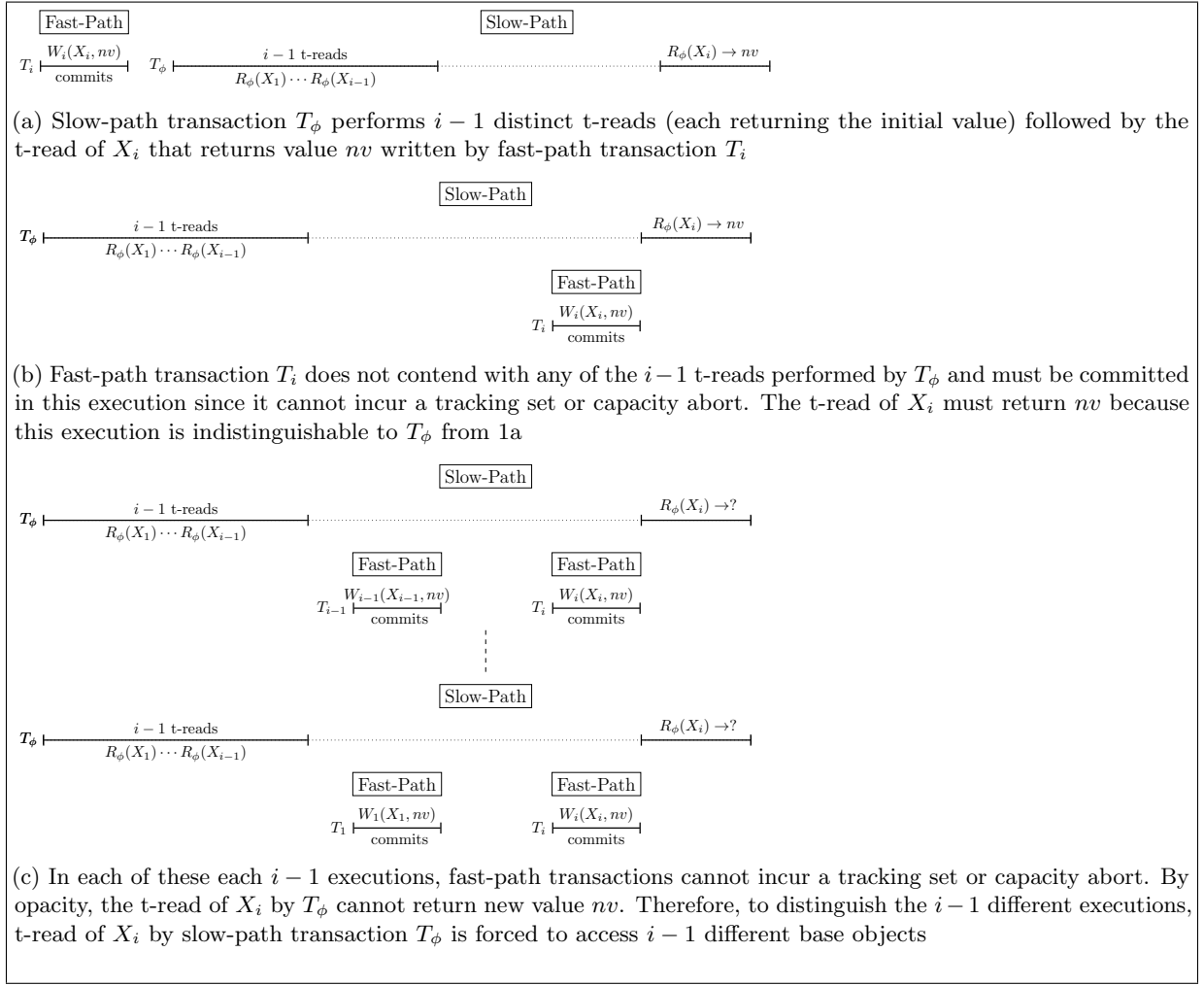

%% file: dap.tex
\begin{tikzpicture}
\node (r1) at (3,0) [] {};
\node (r2) at (12.2,0) [] {};

\node (w1) at (-2,0) [] {};

\draw (r1) node [below] {\normalsize {$R_{\phi}(X_1) \cdots R_{\phi}(X_{i-1})$}};
\draw (r1) node [above] {\normalsize {$i-1$ t-reads}};

\draw (r2) node [above] {\normalsize {$R_{\phi}(X_i)\rightarrow nv$}};

\draw (w1) node [above] {\normalsize {$W_i(X_i,nv)$}}; 
\draw (w1) node [below] {\normalsize {commits}};

\node[draw,align=left] at (8,1) {{\large Slow-Path}};
\node[draw,align=left] at (-2,1) {{\large Fast-Path}};

\begin{scope}   
\draw [|-|,thick] (0,0) node[left] {$T_{\phi}$} to (6,0);
\draw [|-|,thick] (11,0) node[left] {} to (13.5,0);
\draw [-,dotted] (0,0) node[left] {} to (13.5,0);
\end{scope}
\begin{scope}   
\draw [|-|,thick] (-3,0) node[left] {$T_i$} to (-1,0);
\end{scope}
\end{tikzpicture}

%% file: dap2.tex
\begin{tikzpicture}
\node (r1) at (3,0) [] {};
\node (r3) at (14.8,0) [] {};


\node (w2) at (12.5,-2) [] {};

\draw (r1) node [below] {\normalsize {$R_{\phi}(X_1) \cdots R_{\phi}(X_{i-1})$}};
\draw (r1) node [above] {\normalsize {$i-1$ t-reads}};

\draw (w2) node [above] {\normalsize {$W_{i}(X_{i},nv)$}}; 
\draw (w2) node [below] {\normalsize {commits}};

\draw (r3) node [above] {\normalsize {$R_{\phi}(X_{i})\rightarrow nv$}};
\node[draw,align=left] at (10,1) {{\large Slow-Path}};
\node[draw,align=left] at (12.5,-1) {{\large Fast-Path}};

\begin{scope}   
\draw [|-|,thick] (0,0) node[left] {$T_{\phi}$} to (6,0);
\draw [|-|,dotted] (0,0) node[left] {$T_{\phi}$} to (16,0);
\draw [|-|,thick] (13.5,0) node[left] {} to (16,0);
\end{scope}
\begin{scope}   
\draw [|-|,thick] (11.5,-2) node[left] {$T_i$} to (13.5,-2);
\end{scope}
\end{tikzpicture}

%% file: dap3.tex
\begin{tikzpicture}
\node (r1) at (3,0) [] {};
\node (r3) at (14.8,0) [] {};
\node (r33) at (14.8,-5) [] {};
\node (r11) at (3,-5) [] {};

\node (w2) at (12.5,-2) [] {};
\node (w3) at (8,-2) [] {};
\node (w4) at (8,-7) [] {};

\node (w22) at (12.5,-7) [] {};

\draw (r1) node [below] {\normalsize {$R_{\phi}(X_1) \cdots R_{\phi}(X_{i-1})$}};
\draw (r1) node [above] {\normalsize {$i-1$ t-reads}};

\draw (w2) node [above] {\normalsize {$W_{i}(X_{i},nv)$}}; 
\draw (w2) node [below] {\normalsize {commits}};

\draw (w22) node [above] {\normalsize {$W_{i}(X_{i},nv)$}}; 
\draw (w22) node [below] {\normalsize {commits}};

\draw (w3) node [above] {\normalsize {$W_{i-1}(X_{i-1},nv)$}}; 
\draw (w3) node [below] {\normalsize { commits}};

\draw (w4) node [above] {\normalsize {$W_{1}(X_{1},nv)$}}; 
\draw (w4) node [below] {\normalsize { commits}};

\draw (r3) node [above] {\normalsize {$R_{\phi}(X_{i})\rightarrow ?$}};
\node[draw,align=left] at (10,1) {{\large Slow-Path}};
\node[draw,align=left] at (12.5,-1) {{\large Fast-Path}};
\node[draw,align=left] at (8,-1) {{\large Fast-Path}};

\node[draw,align=left] at (8,-6) {{\large Fast-Path}};
\node[draw,align=left] at (12.5,-6) {{\large Fast-Path}};
\begin{scope}   
\draw [|-|,thick] (0,0) node[left] {$T_{\phi}$} to (6,0);
\draw [|-|,dotted] (0,0) node[left] {$T_{\phi}$} to (16,0);
\draw [|-|,thick] (13.5,0) node[left] {} to (16,0);
\end{scope}
\begin{scope}   
\draw (r33) node [above] {\normalsize {$R_{\phi}(X_{i})\rightarrow ?$}};
\node[draw,align=left] at (10,-4) {{\large Slow-Path}};

\draw (r11) node [below] {\normalsize {$R_{\phi}(X_1) \cdots R_{\phi}(X_{i-1})$}};
\draw (r11) node [above] {\normalsize {$i-1$ t-reads}};

\draw [|-|,thick] (0,-5) node[left] {$T_{\phi}$} to (6,-5);
\draw [|-|,dotted] (0,-5) node[left] {$T_{\phi}$} to (16,-5);
\draw [|-|,thick] (13.5,-5) node[left] {} to (16,-5);
\end{scope}

\begin{scope}   
\draw [|-|,thick] (7,-2) node[left] {$T_{i-1}$} to (9,-2);
\draw [|-|,thick] (11.5,-2) node[left] {$T_i$} to (13.5,-2);
\draw [-,dashed] (10,-2.5)  to (10,-3.5);
\draw [|-|,thick] (7,-7) node[left] {$T_1$} to (9,-7);
\draw [|-|,thick] (11.5,-7) node[left] {$T_i$} to (13.5,-7);

\end{scope}
\end{tikzpicture}

%% file: algos.tex
\vspace{-1mm}
\section{Hybrid transactional memory algorithms}\label{sec:hytmalgos}
\begin{figure*}[!t]
      
     \scalebox{.64}[.64]{
     \begin{tabularx}{\textwidth}{c|c|c|c|c}
	~~~~~ & Algorithm~\ref{alg:inswrite} & Algorithm~\ref{alg:inswrite2} & TLE & HybridNOrec\\ 
	Instrumentation in fast-path reads & per-read & constant & constant & constant \\ 
	Instrumentation in fast-path writes & per-write & per-write & constant & constant \\ 
	Validation in slow-path reads & $\Theta(|Rset|)$ & $O(|Rset|)$ & none & $O(|Rset|)$, but validation only if concurrency \\ 
	h/w-s/f concurrency & prog. & prog. for slow-path readers & zero & not prog., but small contention window \\ 
	Direct accesses inside fast-path & yes & no & no & yes \\ 
	opacity & yes & yes & yes & yes 
   \end{tabularx}
\caption{Table summarizing complexities of HyTM implementations}\label{fig:main}    
}
\end{figure*}
%
\input{algos1}
%

%% file: algos1.tex
\noindent\textbf{Instrumentation-optimal progressive HyTM.}
We describe a HyTM algorithm that is a tight bound for Theorem~\ref{th:impossibility} and the instrumentation cost on the fast-path transactions established in \cite{hytm14disc}.
Pseudocode appears in Algorithm~\ref{alg:inswrite}.
For every t-object $X_j$, our implementation maintains a base object $v_j$ that stores the value of $X_j$
and a \emph{sequence lock} $r_{j}$. 

\vspace{1mm}\noindent\textit{Fast-path transactions:}
For a fast-path transaction $T_k$ executed by process $p_i$, the $\Read_k(X_j)$ implementation first reads $r_j$ (direct)
and returns $A_k$ if some other process $p_j$ holds a lock on $X_j$.
Otherwise, it returns the value of $X_j$.
As with $\Read_k(X_j)$, the $\Write (X_j,v)$ implementation returns $A_k$ if some other process $p_j$ holds a lock on $X_j$; otherwise
process $p_i$ increments the sequence lock $r_j$. If the cache has not been invalidated, $p_i$ updates the shared memory
during $\TryC_k$ by invoking the $\ms{commit-cache}$ primitive.

\vspace{1mm}\noindent\textit{Slow-path read-only transactions:}
Any $\Read_k(X_j)$ invoked by a slow-path transaction first reads the value of the t-object from $v_j$, 
adds $r_j$ to $\Rset(T_k)$ if its not held by a concurrent transaction
and then performs \emph{validation} on its entire read set to check if any of them have been modified. 
If either of these conditions is true,
the transaction returns $A_k$. Otherwise, it returns the value of $X_j$. 
Validation of the read set is performed by re-reading the values of the sequence lock entries stored in $\Rset(T_k)$.

\vspace{1mm}\noindent\textit{Slow-path updating transactions:}
An updating slow-path transaction $T_k$ attempts to obtain exclusive write access to its 
entire write set.
If all the locks on the write set were acquired successfully, $T_k$ performs validation of the read set and if successful, updates the values of
the t-objects in shared memory, releases the locks and returns $C_k$; else $p_i$ aborts the transaction.

\vspace{1mm}\noindent\textit{Direct accesses inside fast-path:}
Note that opacity is not violated even if the accesses of the sequence lock during t-read may be performed directly without incurring tracking set aborts.

\vspace{1mm}\noindent\textbf{Instrumentation-optimal HyTM that is progressive only for slow-path reading transactions.}
Algorithm~\ref{alg:inswrite2} does not incur the linear instrumentation cost
on the fast-path reading transactions (inherent to Algorithm~\ref{alg:inswrite}), but provides progressiveness only
for slow-path reading transactions. 
The instrumentation cost on fast-path t-reads is avoided by using a global lock that serializes all updating slow-path transactions
during the $\TryC_k$ procedure. Fast-path transactions simply check if this lock is held without acquiring it (similar to TLE~\cite{tle}). While the per-read instrumentation overhead
is avoided, Algorithm~\ref{alg:inswrite2} still incurs the per-write instrumentation cost.

\vspace{1mm}\noindent\textbf{Sacrificing progressiveness and minimizing contention window.}
Observe that the lower bound in Theorem~\ref{th:impossibility} assumes progressiveness for both slow-path and fast-path transactions
along with opacity and invisible reads. Note that Algorithm~\ref{alg:inswrite2} retains the validation step complexity cost since it provides progressiveness for slow-path readers.

Hybrid NOrec~\cite{hybridnorec} is a HyTM implementation that does not satisfy progressiveness
(unlike its STM counterpart NOrec), but mitigates
the step-complexity cost on slow-path transactions by performing incremental validation 
during a transactional read \emph{iff} 
the shared memory has changed since the start of the transaction.
Conceptually, Hybrid NOrec uses a global sequence lock \emph{gsl} that is incremented 
at the start and end of each transaction's commit procedure.
Readers can use the value of gsl to determine whether shared memory has changed between two configurations.
Unfortunately, with this approach, two fast path transactions will always conflict on the gsl if their 
commit procedures are concurrent.
To reduce the contention window for fast path transactions, the gsl is actually implemented as two separate locks (the second one called \emph{esl}).
A slow-path transaction locks both esl and gsl while it is committing.
Instead of incrementing gsl, a fast path transaction checks if esl is locked and aborts if it is.
Then, at the end of the fast path transaction's commit procedure, 
it increments gsl twice (quickly locking and releasing it and immediately commits in hardware).
Although the window for fast path transactions to contend on gsl is small, our experiments have shown that contention on gsl has a significant impact on performance.

%
\input{prog-fp-sp}
\input{prog-fp}

%% file: prog-fp-sp.tex
\begin{algorithm*}[!t]
\caption{Progressive fast-path and slow-path opaque HyTM implementation; code for transaction $T_k$}
\label{alg:inswrite}
\vspace{-3mm}
\noindent\lstset{style=customc}
\begin{lstlisting}[frame=none,firstnumber=1,mathescape=true]
//\textbf{Shared objects}
    v$_j$, value of each t-object X$_j$ 
    r$_{j}$, a sequence lock of each t-object X$_j$

//\textbf{Code for fast-path transactions}

read$_k(X_j)$
    ov$_j$ := v$_j$  //\label{line:lin1}
    or$_j$ := r$_j$//\medcom direct read\label{line:hread}
    if or$_j.\lit{isLocked}()$ then return $A_k$ 
    return ov$_j$

write$_k(X_j,v)$
    or$_j$ := r$_j$  //\label{line:m1}
    if or$_j.\lit{isLocked}()$ then return $A_k$
    r$_j$ := or$_j.\lit{IncSequence}()$  //\label{line:m2}
    v$_j$ := v  //\label{line:lin2} 
    return OK

tryC$_k$()
    commit-cache$_i$ // \label{line:lin3}

Function: release(Q)
    for each X$_j$ $\in Q$ do r$_j$ := or$_j.\lit{unlock}()$ // \label{line:rel1}
	
Function: acquire(Q)
    for each X$_j$ $\in Q$	
	if r$_j.\lit{tryLock}()$  // \medcom CAS/LLSC \label{line:acq1}
	    Lset(T$_k$) := Lset(T$_k$) $\cup$ {X$_j$}
	else
	    release(Lset(T$_k$))
	    return false
    return true
//%\end{lstlisting}

//\textbf{Code for slow-path transactions}

Read$_k$(X$_j$)
    if X$_j$ $\in$ Wset(T$_k$) then	return Wset(T$_k$).locate(X$_j$) 
    or$_j$ := r$_j$ //\label{line:readorec}
    ov$_j$ := v$_j$ //\label{line:read2}
    Rset(T$_k$) := Rset(T$_k$) $\cup$ {X$_j$,or$_j$} //\label{line:rset}
    if or$_j.\lit{isLocked}()$ then return $A_k$ //\label{line:abort0}	
    if not validate() then return $A_k$ //\label{line:valid}
    return ov$_j$

write$_k(X_j,v)$
    or$_j$ := r$_j$
    nv$_j$ := v
    if or$_j.\lit{isLocked}()$ then return $A_k$
    Wset(T$_k$) := Wset(T$_k$) $\cup$ {X$_j$, nv$_j$, or$_j$}
    return OK

tryC$_k$()
    if Wset(T$_k$) = $\emptyset$ then return $C_k$ // \label{line:return}
    if not acquire(Wset(T$_k$)) then return $A_k$ // \label{line:acq}
    if not validate() // \label{line:abort3}
	    release(Wset(T$_k$)) 
	    return A$_k$ 
    for each X$_j$ $\in$ Wset(T$_k$) do v$_j$:= nv$_j$ //\label{line:write}
    release(Wset(T$_k$))  // \label{line:rel}	
    return C$_k$

Function: validate()
    if $\exists$ X$_j$ $\in$ Rset(T$_k$):or$_j.\lit{getSequence()} \neq$ r$_j.\lit{getSequence()}$ then return false//\label{line:valid}
    return true

\end{lstlisting}
\vspace{-2mm}
\end{algorithm*}

%% file: prog-fp.tex
\begin{algorithm*}[!t]
\caption{Opaque HyTM implementation that is progressive only for slow-path reading transactions; code for $T_k$ by process $p_i$}
\label{alg:inswrite2}
\vspace{-2mm}
\noindent\lstset{style=customc}
\begin{minipage}{0.43\textwidth}
\begin{lstlisting}[frame=none,firstnumber=1,mathescape=true]
//\textbf{Shared objects}
    L, global lock

//\textbf{Code for fast-path transactions}
start$_k$()
    if L$.\lit{isLocked()}$ then return $A_k$

read$_k$(X$_j$)
    ov$_j$ := v$_j$ 
    return ov$_j$

write$_k$(X$_j$, v)
    or$_j$ := r$_j$ 
    r$_j$ := or$_j.\lit{IncSequence}()$ 
    v$_j$ := v 
    return OK

try$C_k$()
    return commit-cache$_i$ 
\end{lstlisting}
\end{minipage}
\hspace{0.02\textwidth}
\begin{minipage}{0.54\textwidth}
\begin{lstlisting}[frame=none,firstnumber=last,mathescape=true]

//\textbf{Code for slow-path transactions}

tryC$_k$()
    if Wset(T$_k$) = $\emptyset$ then return $C_k$
    L$.\lit{Lock}()$
    if not acquire(Wset(T$_k$)) then return $A_k$
    if not validate() then
        release(Wset(T$_k$))
        return $A_k$
    for each X$_j \in$ Wset(T$_k$) do v$_j$ := nv$_j$
    release(Wset(T$_k$))
    return C$_k$
    
Function: release(Q)
    for each X$_j \in$ Q do r$_j$ := nr$_j.\lit{unlock}()$
    L$.\lit{unlock}()$; return OK
\end{lstlisting}
\end{minipage}
\vspace{-2mm}
\end{algorithm*}

%% file: Evaluation.tex
\section{Evaluation}
\label{sec:eval}
In this section, we study the performance characteristics of Algorithms~\ref{alg:inswrite} and \ref{alg:inswrite2}, Hybrid NOrec, TLE and TL2.
Our experimental goals are: (G1) to study the performance impact of instrumentation on the fast-path and validation on the slow-path, 
(G2) to understand how HyTM algorithm design affects performance with Intel and IBM POWER8 HTMs, and 
(G3) to determine whether direct accesses can be used to obtain performance improvements on IBM POWER8 using the supported suspend/resume instruction to escape from a hardware transaction.

\vspace{1mm}\noindent\textbf{Experimental system (Intel).}
The experimental system is a 2-socket Intel E7-4830 v3 with 12 cores per socket and 2 hyperthreads (HTs) per core, for a total of 48 threads.
Each core has a private 32KB L1 cache and 256KB L2 cache (shared between HTs on a core).
All cores on a socket share a 30MB L3 cache.
This system has a non-uniform memory architecture (NUMA) in which threads have significantly different access costs to different parts of memory depending on which processor they are currently executing on.
The machine has 128GB of RAM, and runs Ubuntu 14.04 LTS.
All code was compiled with the GNU C++ compiler (G++) 4.8.4 with build target x86\_64-linux-gnu and compilation options \texttt{-std=c++0x -O3 -mx32}.

We pin threads so that the first socket is saturated before we place any threads on the second socket.
Thus, thread counts 1-24 run on a single socket.
Furthermore, hyperthreading is engaged on the first socket for thread counts 13-24, and on the second socket for thread counts 37-48.
Consequently, our graphs clearly show the effects of NUMA and hyperthreading.

\vspace{1mm}\noindent\textbf{Experimental system (IBM POWER8).}
The experimental system is a IBM S822L with 2x 12-core 3.02GHz processor cards, 128GB of RAM, running Ubuntu 16.04 LTS.
All code was compiled using G++ 5.3.1.
This is a dual socket machine, and each socket has two NUMA \emph{zones}.
It is expensive to access memory on a different NUMA zone, and even more expensive if the NUMA zone is on a different socket.
POWER8 uses the L2 cache for detecting tracking set aborts, and limits the size of a transaction's read- and write-set to 8KB each~\cite{htm-survey}.
This is in contrast to Intel which tracks conflicts on the entire L3 cache, and only limits a transaction's read-set to the L3 cache size, and its write-set to the L1 cache size.

We pin one thread on each core within a NUMA zone before moving to the next zone.
We remark that unlike the thread pinning policy for Intel which saturated the first socket before moving to the next, this proved to be the best policy
for POWER8 which experiences severe negative scaling when threads are saturated on a single 8-way hardware multi-threaded core.
This is because all threads on a core share resources, including the L1 and L2 cache, a single branch execution pipeline, 
and only two load-store pipelines.

\vspace{1mm}\noindent\textbf{Hybrid TM implementations.}
For TL2, we used the implementation published by its authors.
We implemented the other algorithms in C++.
Each hybrid TM algorithm first attempts to execute a transaction on the fast-path, and will continue to execute on the fast-path until the transaction has experienced 20 aborts, at which point it will fall back to the slow-path.
We implemented Algorithm~\ref{alg:inswrite} on POWER8 where each read of a sequence lock during a transactional read operation was enclosed within a pair of suspend/resume instructions to access them without 
incurring tracking set aborts (Algorithm~\ref{alg:inswrite}\textsuperscript{$\ast$}). We remark that this does not affect the opacity of the implementation. 
We also implemented the variant of Hybrid NOrec (Hybrid NOrec\textsuperscript{$\ast$}) in which the update to gsl is performed using a fetch-increment primitive between suspend/resume instructions, as is recommended in~\cite{hynorecriegel}.

In each algorithm, instead of placing a lock next to each address in memory, we allocated a global array of one million locks, and used a simple hash function to map each address to one of these locks.
This avoids the problem of having to change a program's memory layout to incorporate locks, and greatly reduces the amount of memory needed to store locks, at the cost of some possible false conflicts since many addresses map to each lock.
Note that the exact same approach was taken by the authors of TL2.

We chose \textit{not} to compile the hybrid TMs as separate libraries, since invoking library functions for each read and write can cause algorithms to incur enormous overhead.
Instead, we compiled each hybrid TM directly into the code that uses it.

\vspace{1mm}\noindent\textbf{Experimental methodology.}
We used a simple unbalanced binary search tree (BST) microbenchmark as a vehicle to study the performance of our implementations.
The BST implements a dictionary, which contains a set of keys, each with an associated value.
For each TM algorithm 
and update rate $U \in \{40, 10, 0\}$, we run six timed \textit{trials} for several thread counts $n$.
Each trial proceeds in two phases: \textit{prefilling} and \textit{measuring}.
In the prefilling phase, $n$ concurrent threads perform 50\% \textit{Insert} and 50\% \textit{Delete} operations on keys drawn uniformly randomly from $[0, 10^5)$ until the size of the tree converges to a steady state (containing approximately $10^5/2$ keys).
Next, the trial enters the measuring phase, during which threads begin counting how many operations they perform.
In this phase, each thread performs $(U/2)$\% \textit{Insert}, $(U/2)$\% \textit{Delete} and $(100-U)$\% \textit{Search} operations, on keys/values drawn uniformly from $[0,10^5)$, for one second.

Uniformly random updates to an unbalanced BST have been proven to yield trees of logarithmic height with high probability.
Thus, in this type of workload, almost all transactions succeed in hardware, and the slow-path is almost never used.
To study performance when transactions regularly run on slow-path, we introduced an operation called a \textit{RangeIncrement} that often fails in hardware and must run on the slow-path.
A \textit{RangeIncrement}$(low, hi)$ atomically increments the values 
associated with each key in the range $[low, hi]$ present in the tree.
Note that a \textit{RangeIncrement} is more likely to experience data 
conflicts and capacity aborts than BST updates, which only modify a single node.

We consider two types of workloads: (W1) all $n$ threads perform \textit{Insert}, \textit{Delete} and \textit{Search}, and (W2) $n-1$ threads perform \textit{Insert}, \textit{Delete} and \textit{Search} and one thread performs only \textit{RangeIncrement} operations.
Figure~\ref{fig-exp-bst} shows the results for both types of workloads.


\begin{figure}
    \centering
    \setlength\tabcolsep{0pt}
\begin{minipage}{0.495\linewidth}
    \centering
    \textbf{2x12-core Intel E7-4830v3}
    \begin{tabular}{m{0.04\linewidth}m{0.48\linewidth}m{0.48\linewidth}}
        &
        \fcolorbox{black!50}{black!20}{\parbox{\dimexpr \linewidth-2\fboxsep-2\fboxrule}{\centering {\footnotesize No threads perform \textit{RangeIncrement} (W1)}}} &
        \fcolorbox{black!50}{black!20}{\parbox{\dimexpr \linewidth-2\fboxsep-2\fboxrule}{\centering {\footnotesize One thread performs \textit{RangeIncrement} (W2)}}}
        \\
        \rotatebox{90}{\small 0\% updates} &
        \includegraphics[width=\linewidth]{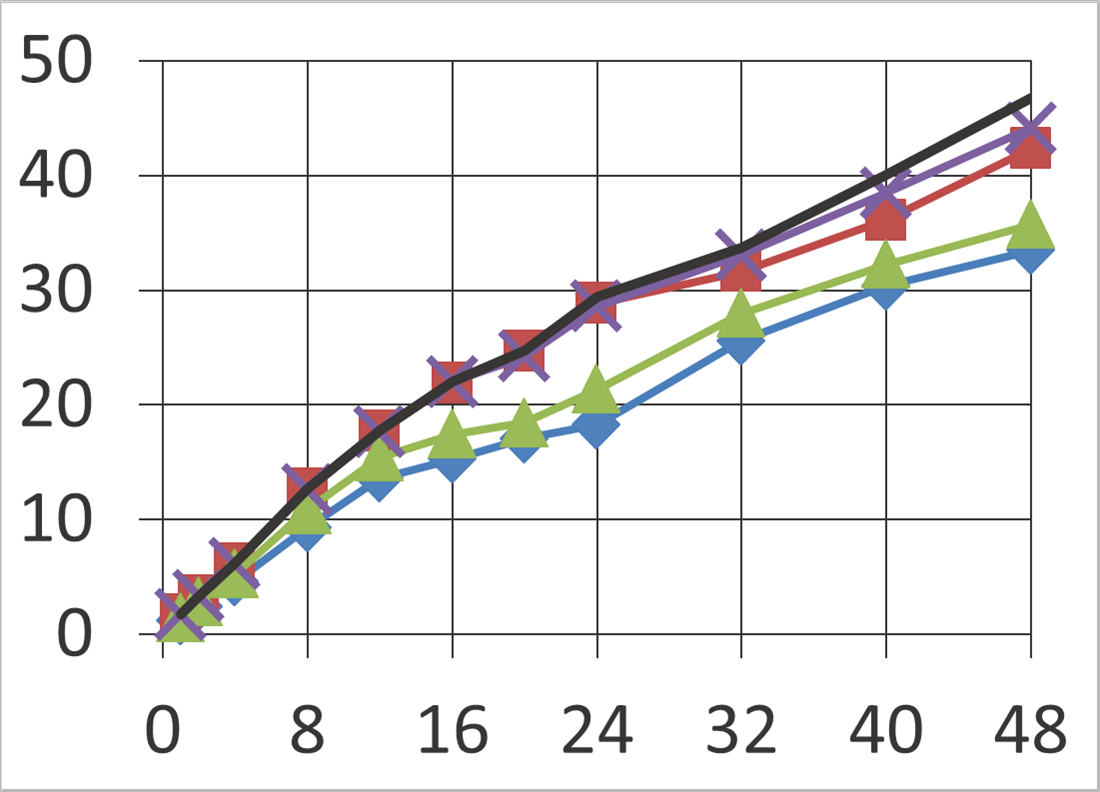} &
        \includegraphics[width=\linewidth]{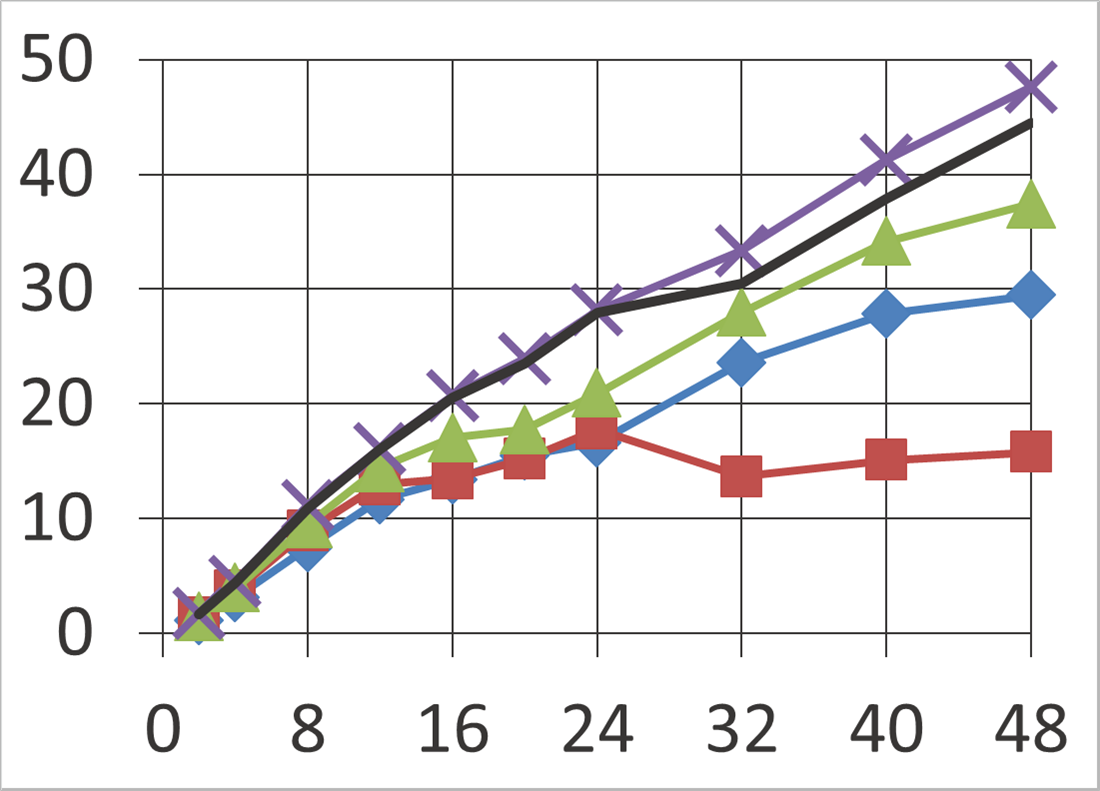}
        \\
        \vspace{-5mm}\rotatebox{90}{\small 10\% updates} &
        \vspace{-5mm}\includegraphics[width=\linewidth]{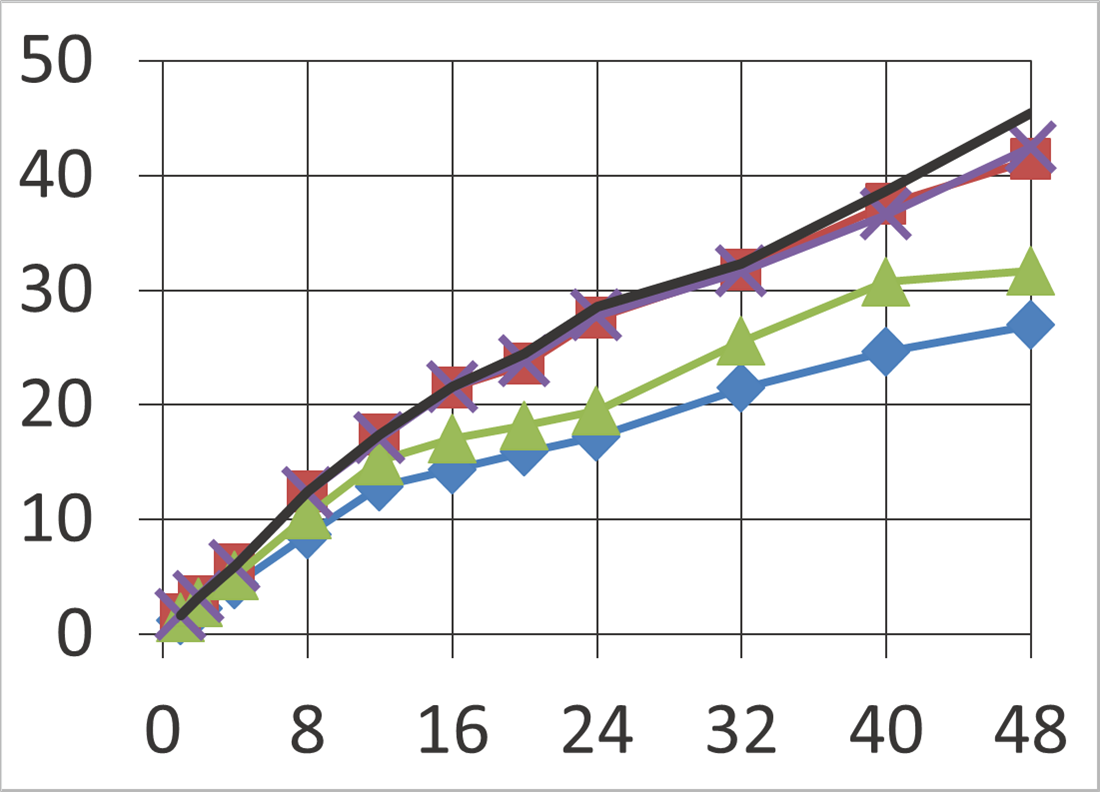} &
        \vspace{-5mm}\includegraphics[width=\linewidth]{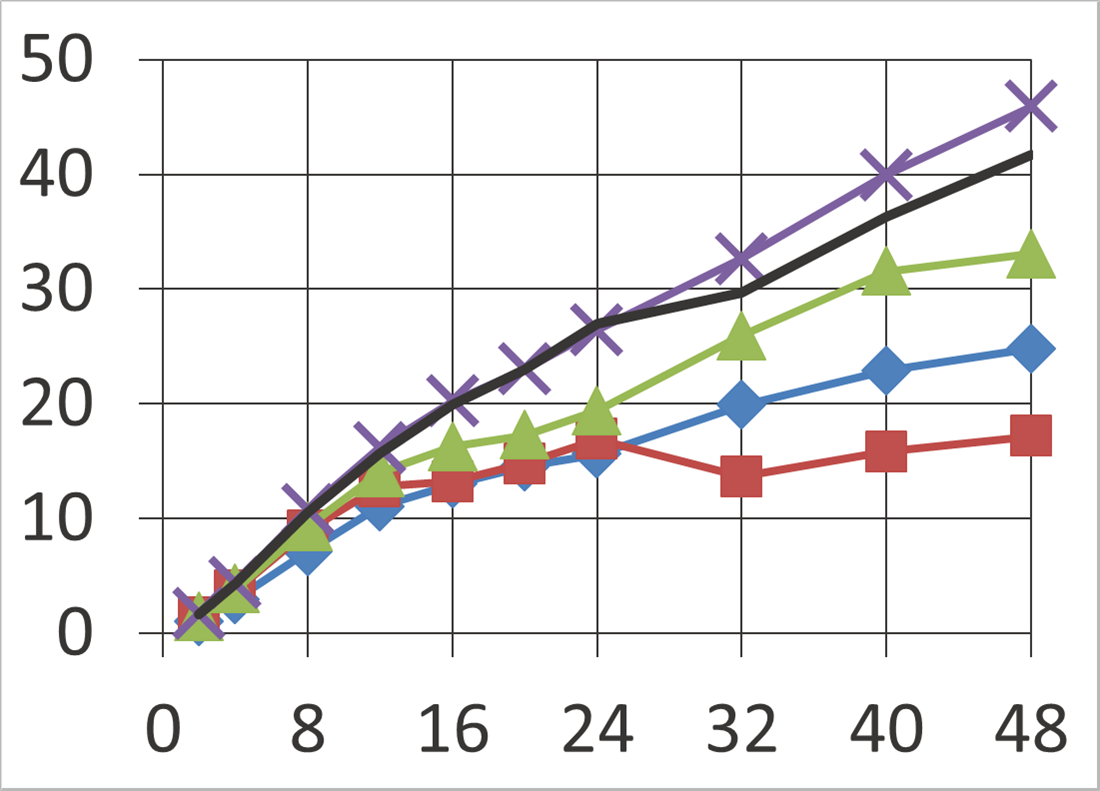}
        \\
        \vspace{-5mm}\rotatebox{90}{\small 40\% updates} &
        \vspace{-5mm}\includegraphics[width=\linewidth]{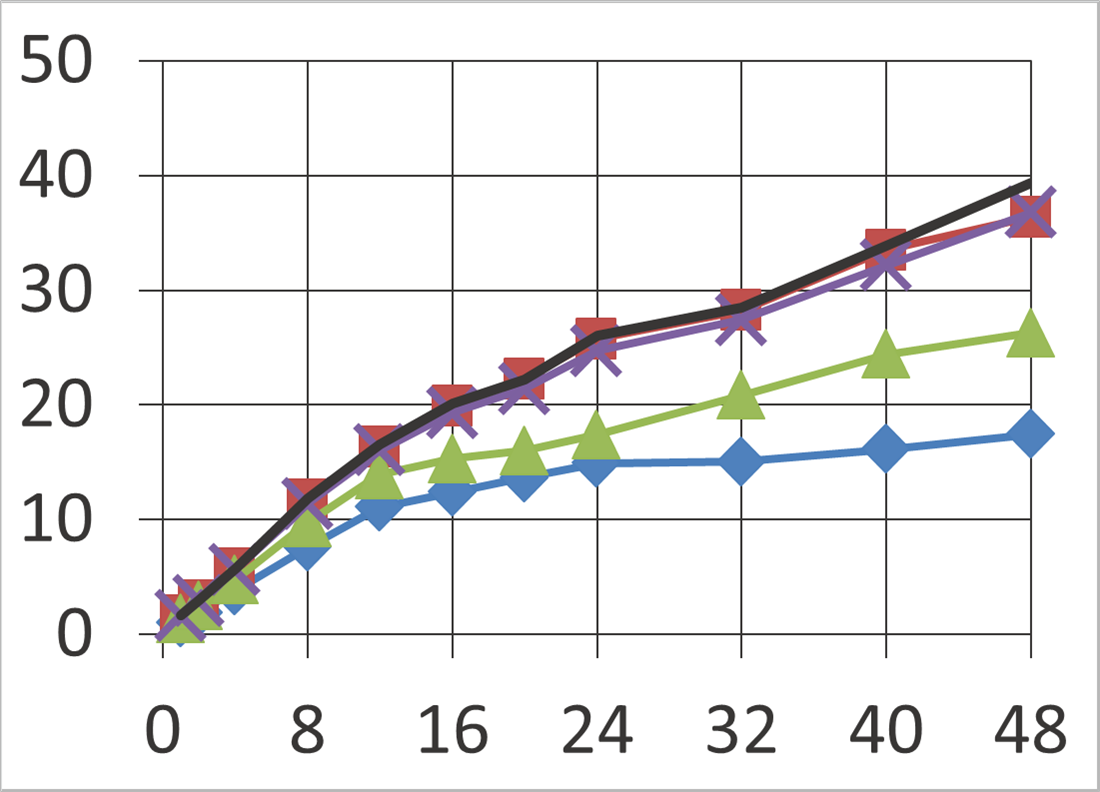} &
        \vspace{-5mm}\includegraphics[width=\linewidth]{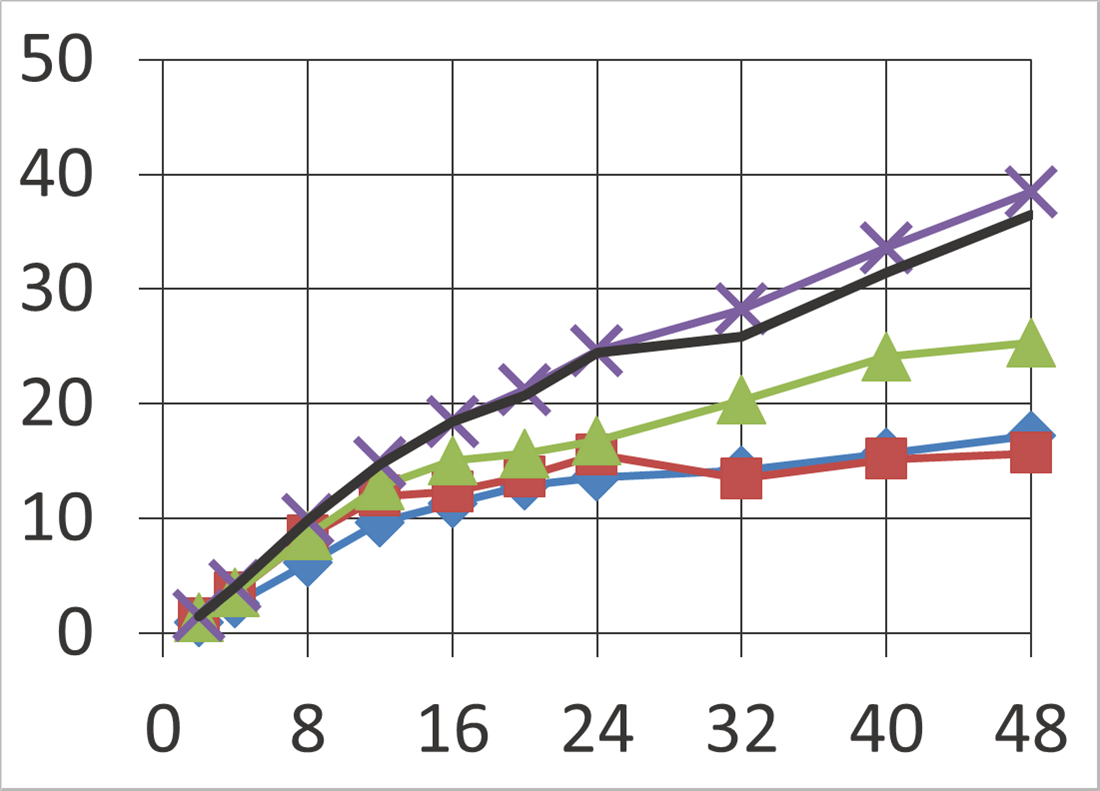}
        \\
    \end{tabular}
\end{minipage}
\begin{minipage}{0.495\linewidth}
    \centering
    \textbf{2x12-core IBM POWER8}
    \begin{tabular}{m{0.04\linewidth}m{0.48\linewidth}m{0.48\linewidth}}
        &
        \fcolorbox{black!50}{black!20}{\parbox{\dimexpr \linewidth-2\fboxsep-2\fboxrule}{\centering {\footnotesize No threads perform \textit{RangeIncrement} (W1)}}} &
        \fcolorbox{black!50}{black!20}{\parbox{\dimexpr \linewidth-2\fboxsep-2\fboxrule}{\centering {\footnotesize One thread performs \textit{RangeIncrement} (W2)}}}
        \\
        \rotatebox{90}{\small 0\% updates} &
        \includegraphics[width=\linewidth]{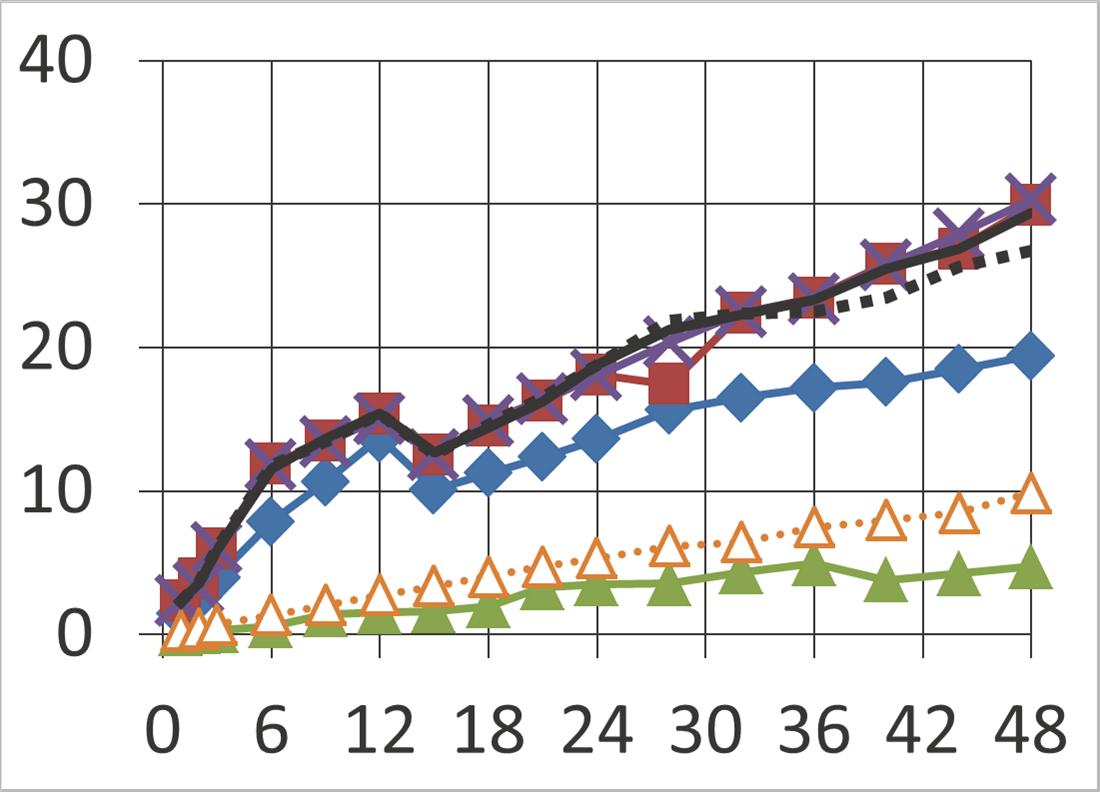} &
        \includegraphics[width=\linewidth]{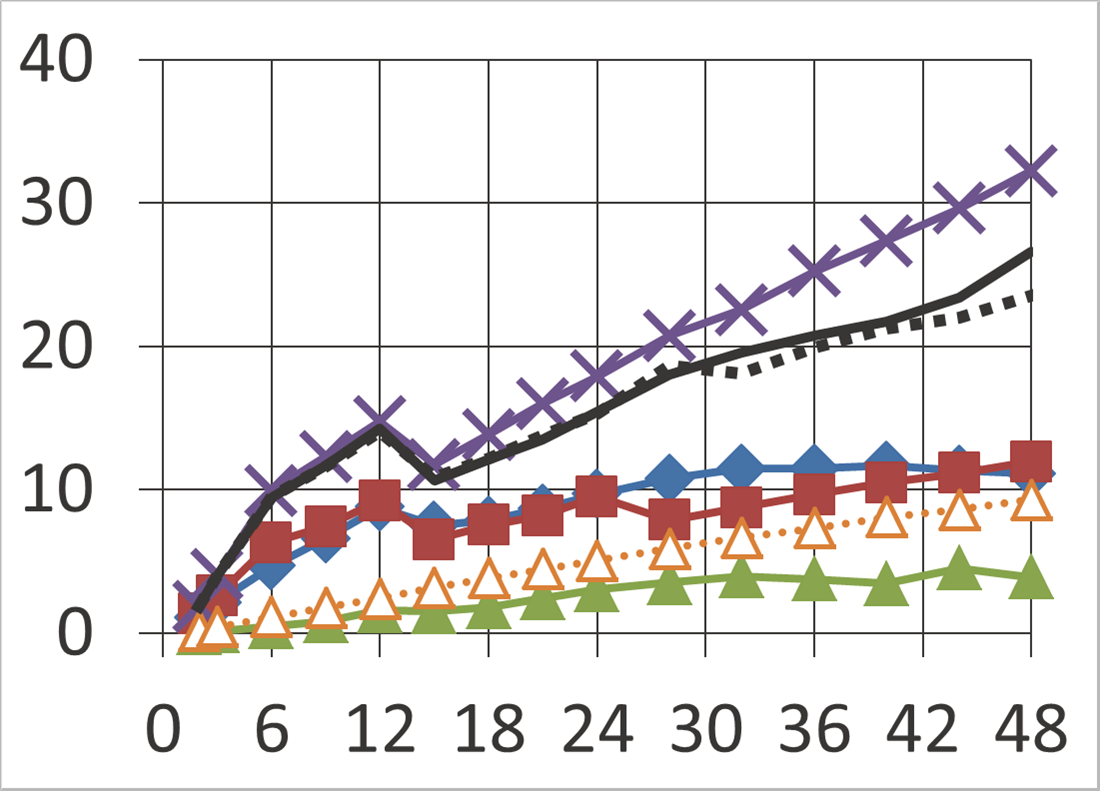}
        \\
        \vspace{-5mm}\rotatebox{90}{\small 10\% updates} &
        \vspace{-5mm}\includegraphics[width=\linewidth]{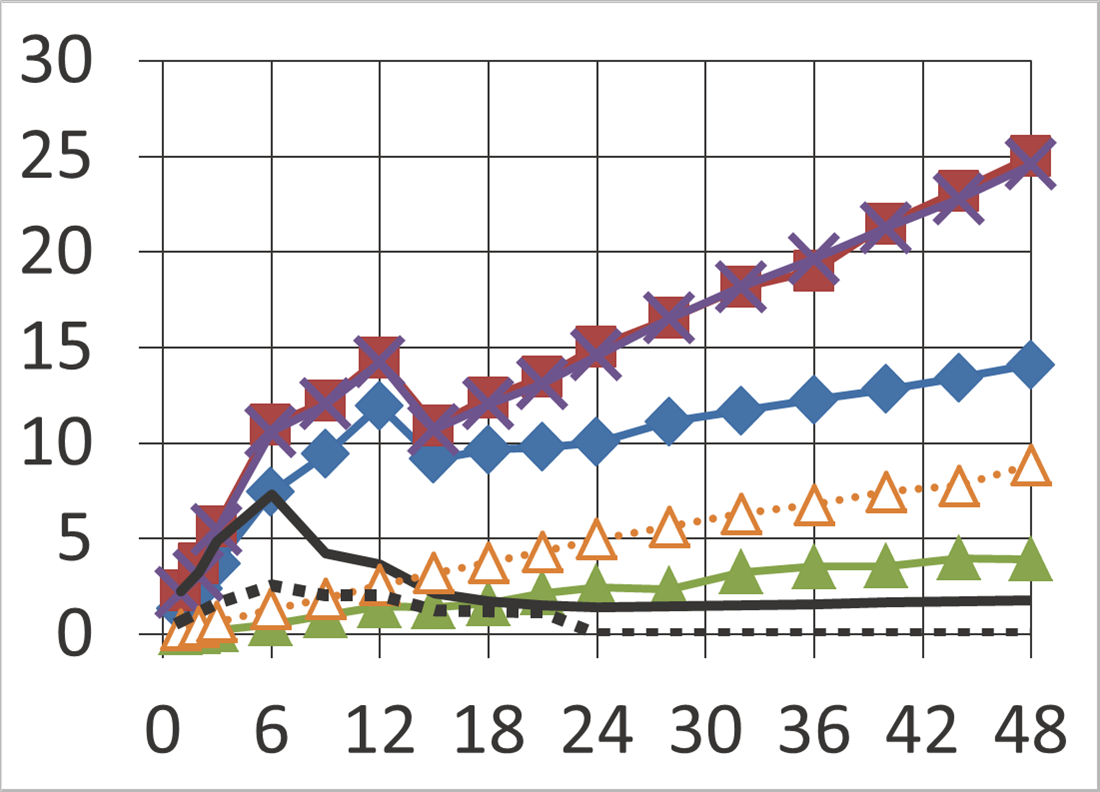} &
        \vspace{-5mm}\includegraphics[width=\linewidth]{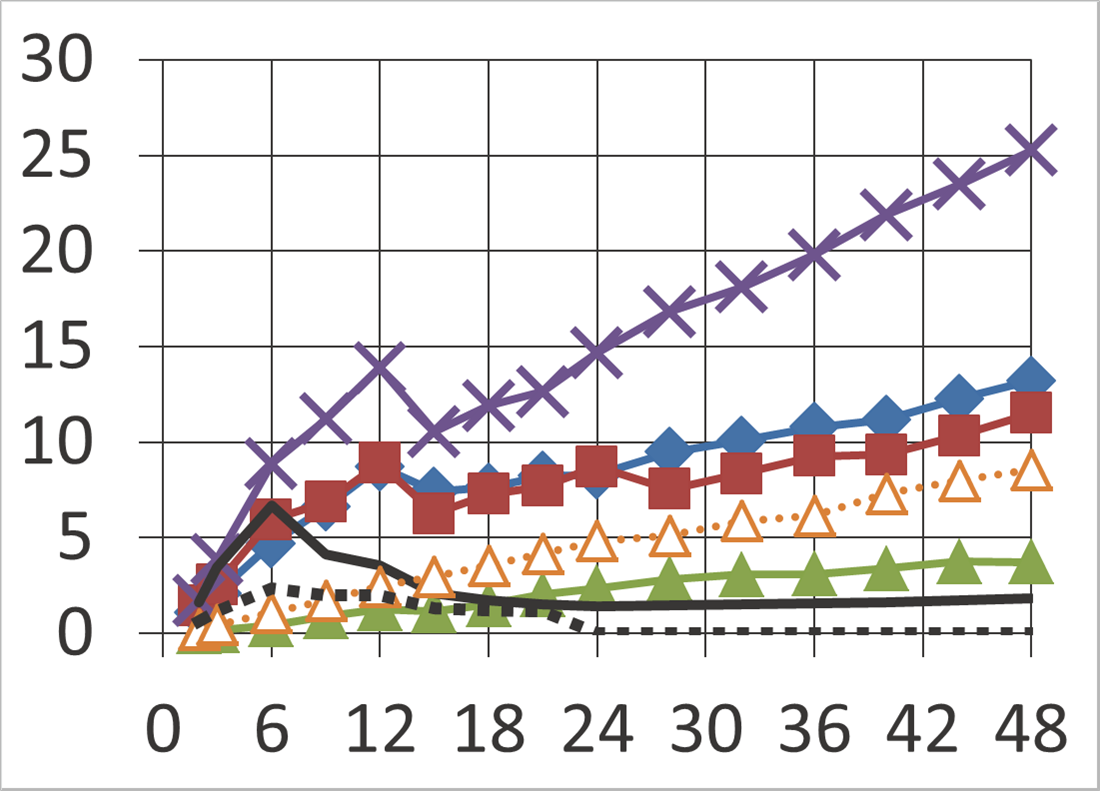}
        \\
        \vspace{-5mm}\rotatebox{90}{\small 40\% updates} &
        \vspace{-5mm}\includegraphics[width=\linewidth]{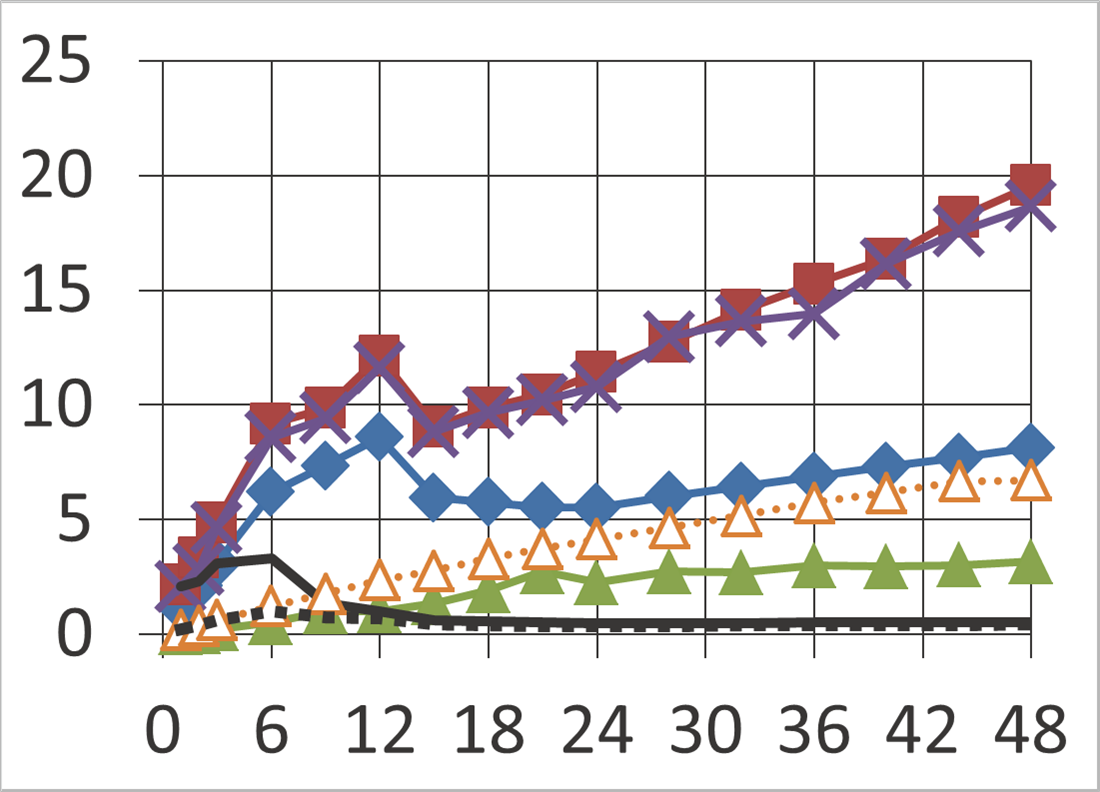} &
        \vspace{-5mm}\includegraphics[width=\linewidth]{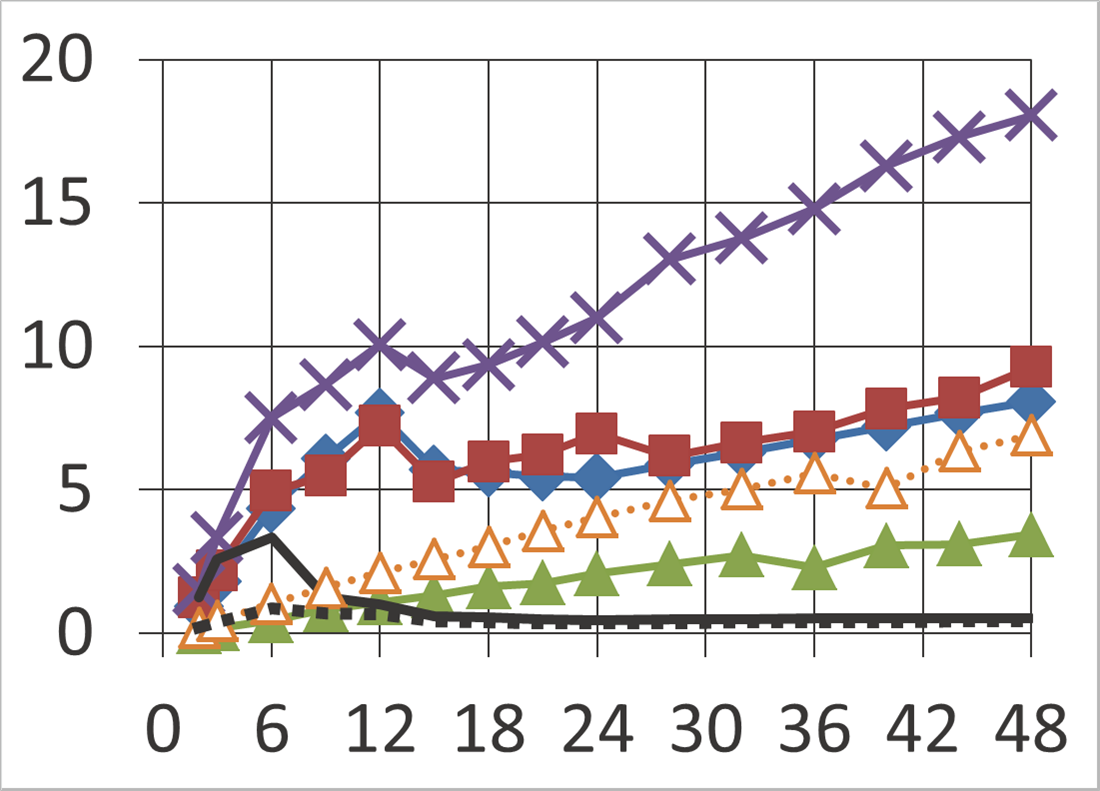}
        \\
    \end{tabular}
\end{minipage}
    \vspace{-2mm}
	\includegraphics[width=\linewidth]{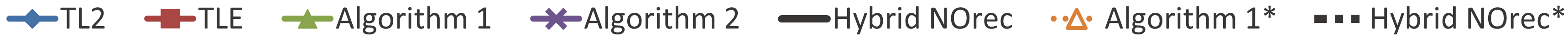}
    \vspace{-4mm}
\caption{Results for a \textbf{BST microbenchmark}.
The x-axis represents the number of concurrent threads.
The y-axis represents operations per microsecond.}
\label{fig-exp-bst}
\end{figure}

\vspace{1mm}\noindent\textbf{Results (Intel).}
We first discuss the 0\% updates graph for workload type W1.
In this graph, essentially all operations committed in hardware.
In fact, in each trial, a small fraction of 1\% of operations ran on the slow-path.
Thus, any performance differences shown in the graph are essentially differences in the performance of the algorithms' respective fast-paths (with the exception of TL2).
Algorithm~\ref{alg:inswrite}, which has instrumentation in its fast-path read operations, has significantly lower performance than Algorithm~\ref{alg:inswrite2}, which does not.
Since this is a read-only workload, this instrumentation is responsible for the performance difference.

In the W1 workloads, TLE, Algorithm~\ref{alg:inswrite2} and Hybrid NOrec perform similarly (with a small performance advantage for Hybrid NOrec at high thread counts).
This is because the fast-paths for these three algorithms have similar amounts of instrumentation: there is no instrumentation for reads or writes, 
and the transaction itself incurs one or two metadata accesses.
In contrast, in the W2 workloads, TLE performs quite poorly, compared to the HyTM algorithms.
In these workloads, transactions must periodically run on the slow-path, and in TLE, 
this entails acquiring a global lock that restricts progress for all other threads.
At high thread counts this significantly impacts performance.
Its performance decreases as the sizes of the ranges passed to \textit{RangeIncrement} increase.
Its performance is also negatively impacted by NUMA effects at thread counts higher than 24.
(This is because, when a thread $p$ reads the lock and incurs a cache miss, 
if the lock was last held by another thread on the same socket, 
then $p$ can fill the cache miss by loading it from the shared L3 cache.
However, if the lock was last held by a thread on a different socket, 
then $p$ must read the lock state from main memory, which is significantly more expensive.)
On the other hand, in each graph in the W2 workloads, the performance of each HyTM (and TL2) is similar to its performance in the corresponding W1 workload graph.
For Algorithm~\ref{alg:inswrite} (and TL2), this is because of progressiveness.
Although Algorithm~\ref{alg:inswrite2} is not truly progressive, fast-path transactions will abort only if they are concurrent with the commit procedure of a slow-path transaction.
In \textit{RangeIncrement} operations, there is a long read-only prefix (which is exceptionally long because of Algorithm~\ref{alg:inswrite2}'s quadratic validation) followed by a relatively small set of writes.
Thus, \textit{RangeIncrement} operations have relatively little impact on the fast-path.
The explanation is similar for Hybrid NOrec (except that it performs less validation than Algorithm~\ref{alg:inswrite2}).

Observe that the performance of Hybrid NOrec decreases slightly, relative to Algorithm~\ref{alg:inswrite2}, after 24 threads.
Recall that, in Hybrid NOrec, the global sequence number is a single point of contention on the fast-path.
(In Algorithm~\ref{alg:inswrite2}, the global lock is only modified by slow-path transactions, so fast-path transactions do not have a single point of contention.)
We believe this is due to NUMA effects, similar to those described in~\cite{BKLL16}.
Specifically, whenever a threads on the first socket performs a fast-path transaction that commits and modifies the global lock, it causes cache invalidations for all other threads.
Threads on socket two must then load the lock state from main memory, which takes much longer than loading it from the shared L3 cache.
This lengthens the transaction's window of contention, making it more likely to abort.
(In the 0\% updates graph in the W2 workload, we still see this effect, because there is a thread performing \textit{RangeIncrement} operations.)

\vspace{1mm}\noindent\textbf{Results (IBM POWER8).}
Algorithm~\ref{alg:inswrite} performs poorly on POWER8: POWER8 transactions can only load 64 cache lines before they will abort~\cite{nguyen-thesis}. 
Transactions read locks and tree nodes, which are in different cache lines: together, they often exceed 64 cache lines loaded in a tree operation, 
so most transactions cannot succeed in hardware. Consequently, on POWER8, 
it is incredibly important either to have minimal instrumentation in transactions, or for metadata to be located in the 
same cache lines as program data. Of course, the latter is not possible for HyTMs, which do not have control over the layout of program data.
Consequently, Algorithm~\ref{alg:inswrite2} outperforms Algorithm~\ref{alg:inswrite} in POWER8 quite easily by avoiding the per-read instrumentation. 

Algorithm~\ref{alg:inswrite} is improved slightly by the expensive (on POWER8) suspend/resume on sequence locks during transactional reads, but it still performs relatively poorly. 
To make suspend/resume a practical tool, one could imagine attempting to 
collect several metadata accesses and perform them together to amortize the cost of a suspend/resume pair. For instance, 
in Algorithm~\ref{alg:inswrite}, one might try to update the locks for all of the transactional writes at once, when the transaction commits. 
Typically one would accomplish this by logging all writes so that a process can remember which addresses it must lock at commit time. 
However, logging the writes inside the transaction would be at least as costly as just performing them.

Observe that Hybrid NOrec does far worse with updates in POWER8 than on the Intel machine.
This is due to the fact that fetch-increment on a single location experiences severe negative scaling on the POWER8 processor: e.g., in one second, a single
thread can perform 37 fetch-add operations while 6 threads perform a total of 9 million and 24 threads perform only 4 million fetch-add operations.
In contrast, the Intel machine performs 32 million operations with 6 threads and 45 million with 24 threads. This is likely because this Intel processor provides 
fetch-add instructions while it must be emulated on the POWER8 processor.

In Hybrid NOrec\textsuperscript{$\ast$}, the non-speculative increment of gsl actually makes performance worse. Recall that in Hybrid NOrec, 
if a fast-path transaction $T_1$ increments gsl, and then a software transaction $T_2$ reads gsl (as part of validation) before $T_1$ commits, then $T_1$ will abort, 
and $T_2$ will not see $T_1$'s change to gsl. 
So, $T_2$ will have a higher chance of avoiding incremental validation (and, hence, will likely take less time to run, and have a smaller contention window).
However, in Hybrid NOrec\textsuperscript{$\ast$}, once $T_1$ increments gsl, $T_2$ will see the change to gsl, regardless of whether $T_1$ commits or aborts. Thus, 
$T_2$ will be forced to perform incremental validation. In our experiments, we observed that a much larger number of transactions ran on 
the fallback path in Hybrid NOrec\textsuperscript{$\ast$} than in Hybrid NOrec (often several orders of magnitude more).

%% file: related.tex
\section{Related work and discussion}
\label{sec:rel}
\vspace{1mm}\noindent\textbf{HyTM implementations and complexity.}
Early HyTMs like the ones described in \cite{damronhytm, kumarhytm} provided progressiveness, but
subsequent HyTM proposals like PhTM~\cite{phasedtm} and HybridNOrec~\cite{hybridnorec} sacrificed progressiveness for lesser instrumentation overheads.
However, the clear trade-off in terms of concurrency vs. instrumentation for these HyTMs have not been studied in the context of currently available HTM
architectures. This instrumentation cost on the fast-path was precisely characterized in \cite{hytm14disc}.
In this paper, we proved the inherent cost of concurrency on the slow-path thus establishing a surprising, 
but intuitive complexity separation between progressive STMs and HyTMs.
Moreover, to the best of our knowledge, this is the first work to consider the theoretical foundations of the cost of concurrency in 
HyTMs in theory and practice (on currently available HTM architectures).
Proof of Theorem~\ref{th:impossibility} is based on the analogous proof for step complexity of STMs that are \emph{disjoint-access parallel}~\cite{prog15-pact, tm-book}.
Our implementation of Hybrid NOrec follows \cite{hynorecriegel}, which additionally proposed the use of direct accesses
in fast-path transactions to reduce instrumentation overhead in the AMD Advanced Synchronization Facility (ASF) architecture.

\vspace{1mm}\noindent\textbf{Beyond the two path HyTM approach.}
\vspace{1mm}\noindent\textit{Employing an uninstrumented fast fast-path.}
We now describe how every transaction may first be executed in a ``fast'' fast-path with almost no instrumentation
and if unsuccessful, may be re-attempted in the fast-path and subsequently in slow-path.
Specifically, we transform any opaque HyTM $\mathcal{M}$ to an opaque
HyTM $\mathcal{M}'$ in which a shared \emph{fetch-and-add} metadata base object $F$ that slow-path updating transactions
increment (and resp. decrement) at the start (and resp. end). In $\mathcal{M}'$, a ``fast'' fast-path transaction checks first if $F$ is $0$
and if not, aborts the transaction; otherwise the transaction is continued as an uninstrumented hardware transaction.
The code for the fast-path and the slow-path is identical to $\mathcal{M}$.

\vspace{1mm}\noindent\textit{Other approaches.}
Recent work has investigated fallback to \emph{reduced} hardware transactions~\cite{MS13}
in which an all-software slow-path is augmented using a slightly faster slow-path 
that is optimistically used to avoid running some transactions on the true software-only slow-path.
Amalgamated lock elision (ALE) was proposed in \cite{ale15} which improves over TLE
by executing the slow-path as a series of segments, each of which is a dynamic length hardware transaction.
Invyswell~\cite{Calciu14} is a HyTM design with multiple hardware and software modes of execution that gives flexibility to avoid instrumentation overhead in uncontended executions.

We remark that such multi-path approaches may be easily applied to each of the Algorithms proposed in this paper. However, 
in the search for an efficient HyTM, it is important to strike the fine balance between concurrency, hardware instrumentation and software validation cost.
Our lower bound, experimental methodology and evaluation of HyTMs provides the first clear characterization of these trade-offs in both Intel and POWER8 architectures. 

%% file: appendix-proof.tex
%
\section{Proof of opacity for algorithms}
\label{app:opacity}
We will prove the opacity of Algorithm~\ref{alg:inswrite} even if some of accesses performed by fast-path transactions are direct (as indicated in the pseudocode).
Analogous arguments apply to Algorithm~\ref{alg:inswrite2}.

Let $E$ by any execution of Algorithm~\ref{alg:inswrite}. 
Since opacity is a safety property, it is sufficient to prove that every finite execution is opaque~\cite{icdcs-opacity}.
Let $<_E$ denote a total-order on events in $E$.

Let $H$ denote a subsequence of $E$ constructed by selecting
\emph{linearization points} of t-operations performed in $E$.
The linearization point of a t-operation $op$, denoted as $\ell_{op}$ is associated with  
a base object event or an event performed during 
the execution of $op$ using the following procedure. 

\vspace{1mm}\noindent\textbf{Completions.}
First, we obtain a completion of $E$ by removing some pending
invocations or adding responses to the remaining pending invocations.
Incomplete $\Read_k$, $\Write_k$ operation performed by a slow-path transaction $T_k$ is removed from $E$;
an incomplete $\TryC_k$ is removed from $E$ if $T_k$ has not performed any write to a base object $r_j$; $X_j \in \Wset(T_k)$
in Line~\ref{line:write}, otherwise it is completed by including $C_k$ after $E$.
Every incomplete $\Read_k$, $\TryA_k$, $\Write_k$ and $\TryC_k$ performed by a fast-path transaction $T_k$ is removed from $E$.

\vspace{1mm}\noindent\textbf{Linearization points.}
Now a linearization $H$ of $E$ is obtained by associating linearization points to
t-operations in the obtained completion of $E$.
For all t-operations performed a slow-path transaction $T_k$, linearization points as assigned as follows:
\begin{itemize}
\item For every t-read $op_k$ that returns a non-A$_k$ value, $\ell_{op_k}$ is chosen as the event in Line~\ref{line:read2}
of Algorithm~\ref{alg:inswrite}, else, $\ell_{op_k}$ is chosen as invocation event of $op_k$
\item For every $op_k=\Write_k$ that returns, $\ell_{op_k}$ is chosen as the invocation event of $op_k$
\item For every $op_k=\TryC_k$ that returns $C_k$ such that $\Wset(T_k)
  \neq \emptyset$, $\ell_{op_k}$ is associated with the first write to a base object performed by $\lit{release}$
  when invoked in Line~\ref{line:rel}, 
  else if $op_k$ returns $A_k$, $\ell_{op_k}$ is associated with the invocation event of $op_k$
\item For every $op_k=\TryC_k$ that returns $C_k$ such that $\Wset(T_k) = \emptyset$, 
$\ell_{op_k}$ is associated with Line~\ref{line:return}
\end{itemize}
For all t-operations performed a fast-path transaction $T_k$, linearization points are assigned as follows:
\begin{itemize}
\item For every t-read $op_k$ that returns a non-A$_k$ value, $\ell_{op_k}$ is chosen as the event in Line~\ref{line:lin1}
of Algorithm~\ref{alg:inswrite}, else, $\ell_{op_k}$ is chosen as invocation event of $op_k$
\item
For every $op_k$ that is a $\TryC_k$, $\ell_{op_k}$ is the $\ms{commit-cache}_k$ primitive invoked by $T_k$
\item
For every $op_k$ that is a $\Write_k$, $\ell_{op_k}$ is the event in Line~\ref{line:lin2}.
\end{itemize}
$<_H$ denotes a total-order on t-operations in the complete sequential history $H$.

\vspace{1mm}\noindent\textbf{Serialization points.}
The serialization of a transaction $T_j$, denoted as $\delta_{T_j}$ is
associated with the linearization point of a t-operation 
performed by the transaction.

We obtain a t-complete history ${\bar H}$ from $H$ as follows. 
A serialization $S$ is obtained by associating serialization points to transactions in ${\bar H}$ as follows:
for every transaction $T_k$ in $H$ that is complete, but not t-complete, 
we insert $\textit{tryC}_k\cdot A_k$ immediately 
after the last event of $T_k$ in $H$. 
If $T_k$ is an updating transaction that commits, then $\delta_{T_k}$ is $\ell_{\TryC_k}$.
If $T_k$ is a read-only or aborted transaction,
then $\delta_{T_k}$ is assigned to the linearization point of the last t-read that returned a non-A$_k$ value in $T_k$.

$<_S$ denotes a total-order on transactions in the t-sequential history $S$.
Since for a given transaction, its
serialization point is chosen between the first and last event of the transaction,
if $T_i \prec_{H} T_j$, then $\delta_{T_i} <_{E} \delta_{T_j}$ implies $T_i <_S T_j$.

Throughout this proof, we consider that process $p_i$ executing fast-path transaction $T_k \in \ms{txns}(E)$
does not include the sequence lock $r_j$ in the tracking set of $p_i$ when accessed in Line~\ref{line:hread}
during $\Read_k(X_j)$.
\begin{claim}
\label{cl:fast}
If every transaction $T_k \in \ms{txns}(E)$ is fast-path, then $S$ is legal.
\end{claim}
\begin{proof}
Recall that Algorithm~\ref{alg:inswrite} performs direct accesses only during the t-read operation in Line~\ref{line:hread} which involves reading the sequence lock $r_j$ corresponding to t-object $X_j$.
However, any two fast-path transactions accessing conflicting data sets must necessarily incur a tracking abort (cf. Remark~\ref{re:traborts}) in $E$. It follows immediately that $S$ must be legal.
\end{proof}
%
%
\begin{claim}
\label{cl:readfrom}
$S$ is legal, i.e., every t-read returns the value of the latest committed t-write in $S$.
\end{claim}
\begin{proof}
We claim that for every $\Read_j(X_m) \rightarrow v$, there exists some slow-path transaction $T_i$ (or resp. fast-path)
that performs $\Write_i(X_m,v)$ and completes the event in Line~\ref{line:write} (or resp. Line~\ref{line:lin2}) such that
$\Read_j(X_m) \not\prec_H^{RT} \Write_i(X_m,v)$.

Suppose that $T_i$ is a slow-path transaction:
since $\Read_j(X_m)$ returns the response $v$, the event in Line~\ref{line:read2}
succeeds the event in Line~\ref{line:write} performed by $\TryC_i$. 
Since $\Read_j(X_m)$ can return a non-abort response only after $T_i$ releases the lock on $r_m$ in
Line~\ref{line:rel1}, $T_i$ must be committed in $S$.
Consequently,
$\ell_{\TryC_i} <_E \ell_{\Read_j(X_m)}$.
Since, for any updating
committing transaction $T_i$, $\delta_{T_i}=\ell_{\TryC_i}$, it follows that
$\delta_{T_{i}} <_E \delta_{T_{j}}$.

Otherwise if $T_i$ is a fast-path transaction, then clearly $T_i$ is a committed transaction in $S$.
Recall that $\Read_j(X_m)$ can read $v$ during the event in Line~\ref{line:read2}
only after $T_i$ applies the $\ms{commit-cache}$ primitive.
By the assignment of linearization points, 
$\ell_{\TryC_i} <_E \ell_{\Read_j(X_m)}$ and thus, $\delta_{T_{i}} <_E \ell_{\Read_j(X_m)}$.

Thus, to prove that $S$ is legal, it suffices to show that  
there does not exist a
transaction $T_k$ that returns $C_k$ in $S$ and performs $\Write_k(X_m,v')$; $v'\neq v$ such that $T_i <_S T_k <_S T_j$. 

$T_i$ and $T_k$ are both updating transactions that commit. Thus, 
($T_i <_S T_k$) $\Longleftrightarrow$ ($\delta_{T_i} <_{E} \delta_{T_k}$) and
($\delta_{T_i} <_{E} \delta_{T_k}$) $\Longleftrightarrow$ ($\ell_{\TryC_i} <_{E} \ell_{\TryC_k}$).

Since, $T_j$ reads the value of $X$ written by $T_i$, one of the following is true:
$\ell_{\TryC_i} <_{E} \ell_{\TryC_k} <_{E} \ell_{\Read_j(X_m)}$ or
$\ell_{\TryC_i} <_{E} \ell_{\Read_j(X_m)} <_{E} \ell_{\TryC_k}$.

Suppose that $\ell_{\TryC_i} <_{E} \ell_{\TryC_k} <_{E} \ell_{\Read_j(X_m)}$.

(\textit{Case \RNum{1}:}) $T_i$ and $T_k$ are slow-path transactions.

Thus, $T_k$ returns a response from the event in Line~\ref{line:acq} 
before the read of the base object associated with $X_m$ by $T_j$ in Line~\ref{line:read2}. 
Since $T_i$ and $T_k$ are both committed in $E$, $T_k$ returns \emph{true} from the event in
Line~\ref{line:acq} only after $T_i$ releases $r_{m}$ in Line~\ref{line:rel1}.

If $T_j$ is a slow-path transaction, 
recall that $\Read_j(X_m)$ checks if $X_j$ is locked by a concurrent transaction, 
then performs read-validation (Line~\ref{line:abort0}) before returning a matching response. 
Indeed, $\Read_j(X_m)$ must return $A_j$ in any such execution.

If $T_j$ is a fast-path transaction, it follows that $\Read_j(X_m)$ must return $A_j$
immediately from Remark~\ref{re:traborts}.

Thus, $\ell_{\TryC_i} <_E \ell_{\Read_j(X)} <_{E} \ell_{\TryC_k}$.

(\textit{Case \RNum{2}:}) $T_i$ is a slow-path transaction and $T_k$ is a fast-path transaction.
Thus, $T_k$ returns $C_k$ 
before the read of the base object associated with $X_m$ by $T_j$ in Line~\ref{line:read2}, but after the response
of \emph{acquire} by $T_i$ in Line~\ref{line:acq}.
Since $\Read_j(X_m)$ reads the value of $X_m$ to be $v$ and not $v'$, $T_i$ performs the \emph{cas}
to $v_m$ in Line~\ref{line:write} after the $T_k$ performs the $\ms{commit-cache}$ primitive (since if
otherwise, $T_k$ would be aborted in $E$).
But then the \emph{cas} on $v_m$ performed by $T_i$ would return $\false$ and $T_i$ would return $A_i$---contradiction.

(\textit{Case \RNum{3}:}) $T_k$ is a slow-path transaction and $T_i$ is a fast-path transaction.
This is analogous to the above case.

(\textit{Case \RNum{4}:}) $T_i$ and $T_k$ are fast-path transactions.
Follows immediately from Claim~\ref{cl:fast}.

We now need to prove that $\delta_{T_{j}}$ indeed precedes $\ell_{\TryC_k}$ in $E$.
Consider the two possible cases.
Suppose that $T_j$ is a read-only transaction. 
Then, $\delta_{T_j}$ is assigned to the last t-read performed by $T_j$ that returns a non-A$_j$ value. 
If $\Read_j(X_m)$ is not the last t-read that returned a non-A$_j$ value, then there exists a $\Read_j(X')$ such that 
$\ell_{\Read_j(X_m)} <_{E} \ell_{\TryC_k} <_E \ell_{read_j(X')}$.
But then this t-read of $X'$ must abort by performing the checks in Line~\ref{line:abort0} or incur a tracking set abort---contradiction.

Otherwise suppose that $T_j$ is an updating transaction that commits, then $\delta_{T_j}=\ell_{\TryC_j}$ which implies that
$\ell_{read_j(X)} <_{E} \ell_{\TryC_k} <_E \ell_{\TryC_j}$. Then, $T_j$ must neccesarily perform the checks
in Line~\ref{line:abort3} and return $A_j$ or incur a tracking set abort---contradiction to the assumption that $T_j$ is a committed transaction.%
\end{proof}
Since $S$ is legal and respects the real-time ordering of transactions, Algorithm~\ref{alg:inswrite} is opaque.

%% file: mproof.tex
\section{Proof of Theorem~\ref{th:impossibility}}
\label{app:lm}
The proof of the lemma below is a simple extension of the analogous lemma from \cite{hytm14disc}
allowing direct trivial accesses inside fast-path transactions which in turn is inspired by an analogous result concerning \emph{disjoint-access parallel} STMs~\cite{AHM09}. 
Intuitively, the proof follows follows from the fact that
the tracking set of a process executing a fast-path transaction is invalidated due to contention on a base
object with another transaction (cf. Remark~\ref{re:traborts}).
\begin{lemma}
\label{lm:hytm}
Let $\mathcal{M}$ be any progressive HyTM implementation in which fast-path transactions may perform trivial
direct accesses.
Let $E_1 \cdot E_2$ be an execution of $\mathcal{M}$ where
$E_1$ (and resp. $E_2$) is the step contention-free
execution fragment of transaction $T_1$ (and resp. $T_2$) executed by process $p_1$ (and resp. $p_2$),
$T_1$ and $T_2$ do not conflict in $E_1 \cdot E_2$, and
at least one of $T_1$ or $T_2$ is a fast-path transaction. 
Then, $T_1$ and $T_2$ do not contend on any base object in $E_1 \cdot E_2$.
\end{lemma}
\begin{proof}
Suppose, by contradiction that $T_1$ and $T_2$ 
contend on the same base object in $E_1\cdot E_2$.

If in $E_1$, $T_1$ performs a nontrivial event on a base object on which they contend, let $e_1$ be the last
event in $E_1$ in which $T_1$ performs such an event to some base object $b$ and $e_2$, the first event
in $E_2$ that accesses $b$ (note that by assumption, $e_1$ is a direct access).
Otherwise, $T_i$ only performs trivial events in $E_1$ to base objects (some of which may be direct) on which it contends with $T_{2}$ in $E_1\cdot E_2$:
let $e_2$ be the first event in $E_2$ in which $E_2$ performs a nontrivial event to some base object $b$
on which they contend and $e_1$, the last event of $E_1$ in $T_1$ that accesses $b$.

Let $E_1'$ (and resp. $E_2'$) be the longest prefix of $E_1$ (and resp. $E_2$) that does not include
$e_1$ (and resp. $e_2$).
Since before accessing $b$, the execution is step contention-free for $T_1$, $E \cdot
E_1'\cdot E_2'$ is an execution of $\mathcal{M}$.
By assumption of lemma, $T_1$ and $T_2$ do not conflict in $E_1'\cdot E_2'$.
By construction, $E_1 \cdot E_2'$ is indistinguishable to $T_2$ from $E_1' \cdot E_2'$.
Hence, $T_1$ and $T_{2}$ are poised to apply contending events $e_1$ and $e_2$ on $b$ in the execution
$\tilde E=E_1' \cdot E_2'$.

We now consider two cases:
\begin{enumerate}
\item 
($e_1$ is a nontrivial event)
After $\tilde E\cdot e_1$, $b$ is contained in the tracking set of process
$p_1$ in exclusive mode and in the extension $\tilde E\cdot e_1 \cdot e_2$, we have that
$\tau_1$ is invalidated. Thus, by Remark~\ref{re:traborts}, transaction $T_1$ must return $A_1$ 
in any extension of $E\cdot e_1\cdot e_2$---a contradiction
to the assumption that $\mathcal{M}$ is progressive.   
\item
($e_1$ is a trivial event)
Recall that $e_1$ may be potentially an event involving a direct access.
Consider the execution $\tilde E\cdot e_2$ following which $b$ is contained in the tracking set of process
$p_{2}$ in exclusive mode. Clearly, we have an extension $\tilde E\cdot e_2 \cdot e_1$ in which
$\tau_{2}$ is invalidated. Thus transaction $T_{2}$ must return $A_{2}$ in any extension of $E\cdot e_2\cdot e_1$---a contradiction
to the assumption that $\mathcal{M}$ is progressive.   
\end{enumerate}
\end{proof}
\begin{theorem}
Let $\mathcal{M}$ be any progressive opaque HyTM implementation providing invisible reads.
There exists an execution $E$ of $\mathcal{M}$ and some slow-path read-only transaction $T_k \in \ms{txns}(E)$
that incurs a time complexity of $\Omega (m^2)$; $m=|\Rset(T_k)|$.
\end{theorem}
\begin{proof}
For all $i\in \{1,\ldots , m\}$; $m \in \mathbb{N}$, let 
$v$ be the initial value of t-object $X_i$.
Let $\pi^{m}$ denote the complete step contention-free execution of a slow-path transaction
$T_{\phi}$ that performs ${m}$ t-reads: $\Read_{\phi}(X_1)\cdots \Read_{\phi}(X_{m})$
such that for all $i\in \{1,\ldots , m \}$, $\Read_{\phi}(X_i) \rightarrow v$.
\begin{claim}
\label{cl:readdap}
For all $i\in \mathbb{N}$, $\mathcal{M}$ has an execution of the form $\pi^{i-1}\cdot \rho^i\cdot \alpha^i$ where,
\begin{itemize}
\item
$\pi^{i-1}$ is the complete step contention-free execution of slow-path read-only transaction $T_{\phi}$ that performs
$(i-1)$ t-reads: $\Read_{\phi}(X_1)\cdots \Read_{\phi}(X_{i-1})$,
\item
$\rho^i$ is the t-complete step contention-free execution of a fast-path transaction $T_{i}$
that writes $nv_i\neq v_i$ to $X_i$ and commits,
\item
$\alpha^i$ is the complete step contention-free execution fragment of $T_{\phi}$ that performs its $i^{th}$ t-read:
$\Read_{\phi}(X_i) \rightarrow nv_i$.
\end{itemize}
\end{claim}
\begin{proof}
$\mathcal{M}$ has an execution of the form $\rho^i\cdot \pi^{i-1}$.
Since $\Dset(T_k) \cap \Dset(T_{i}) =\emptyset$ in $\rho^i\cdot \pi^{i-1}$,
by Lemma~\ref{lm:hytm}, transactions $T_{\phi}$ and $T_i$ do not contend
on any base object in execution $\rho^i\cdot \pi^{i-1}$.
Moreover, since they each access a single t-object, fast-path transaction $T_i$ cannot incur a capacity abort.
Thus, $\rho^i\cdot \pi^{i-1}$ is also an execution of $\mathcal{M}$.

By opacity, $\rho^i\cdot \pi^{i-1} \cdot \alpha^i$ (Figure~\ref{sfig:inv-1}) is an execution
of $\mathcal{M}$ in which the t-read of $X_i$ performed by $T_{\phi}$ must return $nv_i$.
But $\rho^i \cdot \pi^{i-1} \cdot \alpha^i$ is indistinguishable to $T_{\phi}$ from
$\pi^{i-1}\cdot \rho^i \cdot \alpha^i$.
Thus, $\mathcal{M}$ has an execution of the form $\pi^{i-1}\cdot \rho^i \cdot \alpha^i$ (Figure~\ref{sfig:inv-2}).
\end{proof}
For each $i\in \{2,\ldots, m\}$, $j\in \{1,2\}$ and $\ell \leq (i-1)$, 
we now define an execution of the form  $\mathbb{E}_{j\ell}^{i}=\pi^{i-1}\cdot \beta^{\ell}\cdot \rho^i \cdot \alpha_j^i$
as follows:
\begin{itemize}
\item
$\beta^{\ell}$ is the t-complete step contention-free execution fragment of a fast-path transaction $T_{\ell}$
that writes $nv_{\ell}\neq v$ to $X_{\ell}$ and commits
\item
$\alpha_1^i$ (and resp. $\alpha_2^i$) is the complete step contention-free execution fragment of 
$\Read_{\phi}(X_i) \rightarrow v$ (and resp. $\Read_{\phi}(X_i) \rightarrow A_{\phi}$).
\end{itemize}
Note that in the execution so defined above, we assume that each fast-path transactions $T_i$ and $T_{\ell}$;$\ell \leq (i-1)$ are executed by distinct processes.
\begin{claim}
\label{cl:ic2}
For all $i\in \{2,\ldots, m\}$ and $\ell \leq (i-1)$, $\mathcal{M}$ has an execution of the form $\mathbb{E}_{1\ell}^{i}$ or 
$\mathbb{E}_{2\ell}^{i}$.
\end{claim}
%
\begin{proof}
Note that by our assumption on capacity aborts, fast-path transactions $T_i$ and $T_{\ell}$ cannot incur capacity aborts in the defined execution.

For all $i \in \{2,\ldots, m\}$, $\pi^{i-1}$
is an execution of $\mathcal{M}$.
By assumption of invisible reads, $T_{{\ell}}$ must be committed in $\pi^{i-1}\cdot \rho^{\ell}$
and $\mathcal{M}$ has an execution of the form $\pi^{i-1}\cdot \beta^{\ell}$.
By the same reasoning, since $T_i$ and $T_{\ell}$ do not have conflicting data sets,
$\mathcal{M}$ has an execution of the form $\pi^{i-1}\cdot\beta^{\ell}\cdot \rho^i$.

Since the configuration after $\pi^{i-1}\cdot\beta^{\ell}\cdot \rho^i$ is quiescent,
$\pi^{i-1}\cdot\beta^{\ell}\cdot \rho^i$ extended with $\Read_{\phi}(X_i)$
must return a matching response.
If $\Read_{\phi}(X_i) \rightarrow v_i$, then clearly $\mathbb{E}_{1}^{i}$
is an execution of $M$ with $T_{\phi}, T_{i-1}, T_i$ being a valid serialization
of transactions.
If $\Read_{\phi}(X_i) \rightarrow A_{\phi}$, the same serialization
justifies an opaque execution.

Suppose by contradiction that there exists an execution of $\mathcal{M}$ such that
$\pi^{i-1}\cdot\beta^{\ell}\cdot \rho^i$ is extended with the complete execution
of $\Read_{\phi}(X_i) \rightarrow r$; $r \not\in \{A_{\phi},v\}$. 
The only plausible case to analyse is when $r=nv$.
Since $\Read_{\phi}(X_i)$ returns the value of $X_i$ updated by $T_i$, 
the only possible serialization for transactions is $T_{\ell}$, $T_i$, $T_{\phi}$; but $\Read_{\phi}(X_{\ell})$
performed by $T_k$ that returns the initial value $v$
is not legal in this serialization---contradiction.
\end{proof}
\begin{claim}
For all $i\in \{2,\ldots, m\}$, $j\in \{1,2\}$ and $\ell \leq (i-1)$, slow-path transaction $T_{\phi}$ must access
$(i-1)$ different base objects during the execution of $\Read_{\phi}(X_i)$ in the execution
$\pi^{i-1}\cdot \beta^{\ell}\cdot \rho^i \cdot \alpha_j^i$.
\end{claim}
\begin{proof}
Consider the $(i-1)$ different executions: 
$\pi^{i-1}\cdot\beta^{1}\cdot \rho^i$, $\ldots$, $\pi^{i-1}\cdot\beta^{i-1}\cdot \rho^i$ (cf. Figure~\ref{sfig:inv-3}).
For all $\ell, \ell' \leq (i-1)$;$\ell' \neq \ell$, 
$\mathcal{M}$ has an execution of the form $\pi^{i-1}\cdot \beta^{\ell}\cdot \rho^i \cdot \beta^{\ell'}$
in which fast-path transactions $T_{\ell}$ (executed by process $p_{\ell}$) and $T_{\ell'}$ (executed by process $p_{\ell'}$) access mutually disjoint data sets.
By invisible reads and Lemma~\ref{lm:hytm}, the pairs of transactions $T_{\ell'}$, $T_{i}$ and $T_{\ell'}$, $T_{\ell}$
do not contend on any base object in this execution.
This implies that $\pi^{i-1}\cdot \beta^{\ell} \cdot \beta^{\ell'} \cdot \rho^i$ is an execution of $\mathcal{M}$ in which
transactions $T_{\ell}$ and $T_{\ell'}$ each apply nontrivial primitives
to mutually disjoint sets of base objects in the execution fragments $\beta^{\ell}$ and $\beta^{\ell'}$ respectively.

This implies that for any $j\in \{1,2\}$, $\ell \leq (i-1)$, the configuration $C^i$ after $E^i$ differs from the configurations
after $\mathbb{E}_{j\ell}^{i}$ only in the states of the base objects that are accessed in the fragment $\beta^{\ell}$.
Consequently, slow-path transaction $T_{\phi}$ must access at least $i-1$ different base objects
in the execution fragment $\pi_j^i$
to distinguish configuration $C^i$ from the configurations
that result after the $(i-1)$ different executions 
$\pi^{i-1}\cdot\beta^{1}\cdot \rho^i$, $\ldots$, $\pi^{i-1}\cdot\beta^{i-1}\cdot \rho^i$ respectively.
\end{proof}
Thus, for all $i \in \{2,\ldots, m\}$, slow-path transaction $T_{\phi}$ must perform at least $i-1$ steps 
while executing the $i^{th}$ t-read in the execution fragment $\pi_{j}^i$. 
Inductively, this gives the $\sum\limits_{i=1}^{m-1} i=\frac{m(m-1)}{2}$ step complexity for $T_{\phi}$ thus completing the proof.
\end{proof}